\RequirePackage{fix-cm}
\documentclass[smallextended]{svjour3}       % onecolumn (second format)
\smartqed  % flush right qed marks, e.g. at end of proof
\usepackage{amssymb}
\usepackage[latin1]{inputenc}
\expandafter\let\csname equation*\endcsname\relax   % necessary for amsmath
\expandafter\let\csname endequation*\endcsname\relax   % necessary for amsmath
\usepackage[reqno]{amsmath}
\usepackage{amstext}
\usepackage{enumerate}
\usepackage{ifthen}
\usepackage{lineno}
\usepackage{tabularx}
\usepackage{booktabs}
\usepackage{multirow}
\usepackage{rotating}
\usepackage[numbers,square,sort&compress]{natbib}

\usepackage{color}

\usepackage[hyphens]{url}
\usepackage{hyperref}
\hypersetup{
  pdftitle={Attacks on quantum key distribution protocols that employ non-ITS authentication},
  pdfsubject={pdfsubject},
  pdfproducer={pdfproducer},
  pdfauthor={Christoph Pacher}
}

\newcommand{\Notation}{\emph{Notation}.\ }
\newcommand{\Setup}{\emph{Setup}.\ }
\newcommand{\Messages}{\emph{Protocol messages}.\ }
\newcommand{\Attackmessages}{\emph{Protocol messages and messages inserted by Eve (marked by ')}.\ }
\newcommand{\Actions}{\emph{Protocol actions}.\ }
\newcommand{\Attackactions}{\emph{Protocol and attack actions}.\ }
\newcommand{\Alice}{\textbf{A}}
\newcommand{\Bob}{\textbf{B}}
\newcommand{\Eve}{\textbf{E}}
\newcommand{\common}[2]{#1=#2}
\newcommand{\Info}{d}
\newcommand{\Bases}{b}
\newcommand{\tildeBases}{\tilde{\Bases}}
\newcommand{\Siftkey}{s}
\newcommand{\tildeSiftkey}{\tilde{\Siftkey}}
\newcommand{\Key}{K}
\newcommand{\state}{\rho}

\newcommand{\IA}{\Info^\Alice}
\newcommand{\IB}{\Info^\Bob}
\newcommand{\IE}{\Info^\Eve}
\newcommand{\BA}{\Bases^\Alice}
\newcommand{\BB}{\Bases^\Bob}

\newcommand{\BE}{\Bases^\Eve}
\newcommand{\tildeBE}{\tildeBases^\Eve}

\newcommand{\SA}{\Siftkey^\Alice}
\newcommand{\hatSA}{\hat{\Siftkey}^\Alice}
\newcommand{\SB}{{\Siftkey}^\Bob}
\newcommand{\hatSB}{\hat{\Siftkey}^\Bob}
\newcommand{\SE}{\Siftkey^\Eve}
\newcommand{\SEA}{\Siftkey^{\Eve \leftrightarrow \Alice}}
\newcommand{\hatSEA}{\hat{\Siftkey}^{\Eve \leftrightarrow \Alice}}
\newcommand{\SEB}{\Siftkey^{\Eve \leftrightarrow \Bob}}
\newcommand{\hatSEB}{\hat{\Siftkey}^{\Eve \leftrightarrow \Bob}}
\newcommand{\tildeSEB}{\tildeSiftkey^{\Eve \leftrightarrow \Bob}}
\newcommand{\KA}{\Key^\Alice}
\newcommand{\KB}{\Key^\Bob}
\newcommand{\KE}{\Key^\Eve}
\newcommand{\KEA}{\Key^{\Eve \leftrightarrow \Alice}}
\newcommand{\KEB}{\Key^{\Eve \leftrightarrow \Bob}}
\newcommand{\sA}{\state^\Alice}
\newcommand{\sE}{\state^\Eve}
\newcommand{\BAB}{{\Bases}^{\common{\Alice}{\Bob}}}
\newcommand{\BEB}{{\Bases}^{\common{\Eve}{\Bob}}}
\newcommand{\tildeBEB}{{\tildeBases}^{\common{\Eve}{\Bob}}}
\newcommand{\BAE}{{\Bases}^{\common{\Alice}{\Eve}}}
\newcommand{\tildeBAE}{{\tildeBases}^{\common{\Alice}{\Eve}}}

\newcommand{\myempty}{\text{\emph{empty}}}

\newcommand{\nonce}{n}
\newcommand{\nB}{\nonce^\Bob}
\newcommand{\nE}{\nonce^\Eve}

\newcommand{\mA}{m^\Alice}
\newcommand{\mE}{m^\Eve}
\newcommand{\tildemE}{\tilde{m}^\Eve}

\newcommand{\TA}{T^\Alice}
\newcommand{\TE}{T^\Eve}

\newcommand{\qchan}{\mathcal{Q}}
\newcommand{\cchan}{\mathcal{C}}

\newcommand{\mack}{\textrm{m}_\textrm{ack}}

\newcommand{\MAC}[2]{g_{#1}(#2)}
\newcommand{\msg}[2][\empty]{
	\ifthenelse{\equal{#1}{\empty}}
		{#2}
		{#2,\MAC{#1}{#2}}
}

\newcommand{\ECCode}{EC}

\newcommand{\COCode}{CO}
\newcommand{\errorr}{\epsilon}
\newcommand{\fail}{\text{\emph{fail}}}

\newcommand{\PACode}{P}

\newcommand{\PA}{\PACode^\Alice}
\newcommand{\PE}{\PACode^\Eve}

\newcommand{\mS}{\mathcal{S}}
\newcommand{\mG}{\mathcal{G}}
\newcommand{\mZ}{\mathcal{Z}}
\newcommand{\mF}{\mathcal{F}}
\newcommand{\mH}{\mathcal{H}}
\newcommand{\mM}{\mathcal{M}}
\newcommand{\mT}{\mathcal{T}}

\newcommand{\affAIT}{1}
\newcommand{\affLiU}{2}
\newcommand{\affIQOQI}{3}
\newcommand{\affVCQ}{4}

\newcommand{\br}{\hline}
\newcommand{\mr}{\hline}

\begin{document}

\title{Attacks on quantum key distribution protocols that employ non-ITS authentication}%\thanks{Grants or other notes

\titlerunning{Attacks on QKD protocols that employ non-ITS authentication}        % if too long for running head

\author{C~Pacher$^{\affAIT}$     \and
        A~Abidin$^{\affLiU}$     \and
        T~Lor\"unser$^{\affAIT}$ \and 
        M~Peev$^{\affAIT}$       \and
        R~Ursin$^{\affIQOQI}$    \and
        A~Zeilinger$^{\affIQOQI,\affVCQ}$\and
        J-{\AA}~Larsson$^{\affLiU}$}

\authorrunning{C~Pacher et al.}

\institute{$^{\affAIT}$ C~Pacher (\email{christoph.pacher@ait.ac.at}), T~Lor\"unser, and M~Peev \at Digital Safety \& Security Department, AIT Austrian Institute of Technology, Austria 
\and
$^{\affLiU}$ A~Abidin and J-{\AA}~Larsson \at Department of Electrical Engineering, Link\"oping University, Link\"oping, Sweden
\and
%
% \address{$^{\affIQOQI}$Institut f\"ur Experimentalphysik, Universit\"at Wien, Boltzmanngasse 5, A-1090 Wien, Austria and Institute for Quantum Optics and Quantum Information, Austrian Academy of Sciences, Boltzmanngasse 3, A-1090 Wien, Austria}
%
$^{\affIQOQI}$ R~Ursin and A~Zeilinger \at Institute for Quantum Optics and Quantum Information, Austrian Academy of Sciences, Austria
\and
$^{\affVCQ}$ A~Zeilinger \at Vienna Center for Quantum Science and Technology (VCQ), Faculty of Physics, University of Vienna, Austria.}


\maketitle

\begin{abstract}
We demonstrate how adversaries with large computing resources can break Quantum Key Distribution (QKD) protocols which employ a particular message authentication code suggested previously. This authentication code, featuring low key consumption,  is not Information-Theoretically Secure (ITS) since for each message the eavesdropper has intercepted she is able to send a different message from a set of messages that she can calculate by finding collisions of a cryptographic hash function. However, when this authentication code was introduced it was shown to prevent straightforward Man-In-The-Middle (MITM) attacks against QKD protocols.

In this paper, we prove that the set of messages that collide with any given message under this authentication code contains with high probability a message that has small Hamming distance to any other given message. Based on this fact we present extended MITM attacks against different versions of BB84 QKD protocols using the addressed authentication code; for three protocols we describe every single action taken by the adversary. For all protocols the adversary can obtain complete knowledge of the key, and for most protocols her success probability in doing so approaches unity. Since the attacks work against all authentication methods which allow to calculate colliding messages, the underlying building blocks of the presented attacks expose the potential pitfalls arising as a consequence of non-ITS authentication in QKD-postprocessing. We propose countermeasures, increasing the eavesdroppers demand for computational power, and also prove necessary and sufficient  conditions for upgrading the discussed authentication code to the ITS level.
\keywords{Quantum Key Distribution \and Information-Theoretic Security \and Message Authentication \and Collision Attacks \and Man-in-the-Middle attack}
\end{abstract}

\section{Introduction}

Quantum key distribution (QKD) is a cryptographic key-agreement protocol.
The principles of QKD and particularly the very first proposal 
(Bennett-Brassard-84 or BB84 \cite{BB84}) are well-known. Still we prefer to start with a 
short outline  in order to provide the reader with a very general picture leaving more  
in-depth descriptions for the technical part of the paper (see  Sections~\ref{sec:prot1}--\ref{sec:prot3}). 

Any version of QKD consists of two phases: quantum communication (and measurements), and classical
post processing. The quantum communication phase typically uses states of
light. In the original approach by Bennett and Brassard polarized photons have been put forward to 
this end.  In the case of BB84 the quantum
communication phase starts with the generation of a random bitsequence, by one of the
users of the system (Alice). This bitsequence is to be used as raw material for the
key. Alice encodes it at random into either a horizontal-vertical polarization basis, e.g. 
horizontal -- 1 and vertical -- 0,  or analogously  into a $\pm45^\circ$ polarization basis. The other user (Bob) 
decodes the bits, selecting randomly either horizontal-vertical polarization or $\pm45^\circ$
polarization as a measurement basis. In case of perfect communication, the bits of Alice and Bob would agree whenever they have chosen identical coding and decoding bases. An eavesdropper (Eve) would need to guess what coding basis is used
(in the simplest \emph{intercept-resend} attack) to intercept Alice's signal in a measurement basis and resend the result to Bob in the same basis. Therefore whenever Alice's and Bob's encoding/decoding choices 
coincide, there is a 50\% chance that Eve guesses wrongly. Thus as consequence of the quantum \emph{uncertainty
  relation} for polarization  she will destroy the information potentially shared by Alice and Bob and render the bitstings of the two-parties uncorrelated.

After the quantum communication phase has ended, Alice and Bob start the
classical post processing through a process called \emph{sifting}. They
compare what encoding/decoding bases they used for each bit, and keep only those bit
values for which these bases match. This will remove 50\% of the raw
key bits on average. The next step is to perform \emph{reconciliation},
or error-correction, so that channel noise can be removed. This step will also
indicate if there is an eavesdropper on the channel. Typically, the users have
agreed on an acceptable noise level, and abort the protocol if the noise level
is too high. Whether Alice's and Bob's keys have been successfully reconciled is determined in the \emph{confirmation} step. 
A subsequent step is \emph{privacy amplification}, sacrificing some of
the newly established key material to decrease exponentially Eve's knowledge on the final
key, the assumption being that any channel noise is due to (some moderate level of) 
eavesdropping by Eve. The final post-processing step is \emph{authentication} to verify the
origin and correctness of the classical communication messages, used in the post processing phase. 
This sequence may differ in
detail and order in different QKD protocols, but in principle all
separate steps need to be  present
in one way or another.

The authentication step is the focus of this paper; the outstanding property
of QKD is that it is an Information-Theoretically Secure (ITS) and universally
composable (UC) protocol given that its classical communication is performed
over an authentic channel (note that \emph{all} key-agreement protocols are
insecure over non-authentic channels).  Thus QKD is a very powerful
cryptographic primitive but in order to be useful for practical key agreement
purposes it must be \emph{composed} with an independent primitive enforcing
the mentioned requirement for authenticity of classical communication.

The standard cryptographic approach ensuring authenticity of
communication messages against malicious attackers is to use a
message authentication code (MAC) \cite{MenezesOV96}.  A convenient
class of MACs are \emph{systematic} MACs which replace the original
message with a concatenation of the message itself and an additional
tag which is the image of a keyed hash function applied to the
message.  It is well-known that Strongly Universal$_2$ (SU$_2$), and
more generally Almost Strongly Universal$_2$ (ASU$_2$) hashing
(see \ref{uhf}) is an ITS primitive that can be used to calculate
systematic MAC tags.

\subsection{Related work}
Very recently authentication based on ASU$_2$ hashing was explicitly
shown \cite{Portmann2012} to be
also UC (a fact that has been used implicitly for quite some time).
Therefore UC message authentication with ASU$_2$ hashing can be
composed with UC quantum key distribution over authentic channels to
form a UC (quantum-classical) key agreement protocol over
non-authentic channels.
Thus, ASU$_2$ hashing is sufficient for the authentication of the
classical messages exchanged during any QKD protocol. However,
although composing two UC primitives is \emph{sufficient} for getting
a UC composed protocol this is not a priori \emph{necessary} as
in principle it is not excluded that it can be shown directly that the final protocol is UC.
In this sense it might still be possible that QKD over
non-authentic channels can be made secure without relying on ASU$_2$
hashing.  Alternatives using weaker authentication have been proposed,
and this paper focuses on the method of Ref.~\cite{Peev2005}, that
puts forward a hash function which is a composition of an (inner) known
public hash function (like SHA) and an (outer) SU$_2$ function.  It
was proven that QKD using this authentication is secure against an
eavesdropper that attempts to break the protocol using a
straightforward "man-in-the-middle" (MITM) attack, as defined below.
Later, in Refs.~\cite{BMQ05,Abidin2009} it was observed that an
eavesdropper can apply more advanced strategies than a straightforward
MITM and get a significant leverage by being able to break QKD with
particular realizations of post-processing. It has, however, been
argued \cite{Peev2005,Peev2009} that this weakness occurs only in
specific post processing realizations, while in practical (or generic)
ones the proposed eavesdropping techniques remain inadequate.

\subsection{New results}
In this paper we use the adversarial approaches of \cite{BMQ05,Abidin2009},
extend them significantly to full scale eavesdropping strategies, and
demonstrate in detail how to break several explicit QKD protocols,
that employ the authentication method of Ref.~\cite{Peev2005}, under
the assumptions that the adversary possesses unbounded computation
resources and in some cases quantum memory.
The general attack-pattern is a
sophisticated (interleaving) MITM attack, in which the adversary (Eve) carries out independent protocols with the legitimate parties (Alice and Bob). In doing so Eve  manages to modify her respective protocol messages such that these collide with those of Alice and Bob  under the
first part of the authentication of Ref.~\cite{Peev2005}.
% \footnote{Technically speaking this is a second-preimage attack against the public hash function which is always feasible given an adversary with unbounded computing resources.  }
Depending
on the protocol variants (e.g., immediate vs.\ delayed
authentication), the different attacks which we study address sifting,
error correction, confirmation, and privacy amplification or only some
of these steps.  These techniques can be used to break a very broad
class of post-processing protocol realizations which include those
routinely used in practical implementations. With significant probability
that in most attacks approaches unity Eve shares
a key with the legitimate parties.

We also consider some countermeasures, which consist of modifications of
the two-step authentication mechanism. These modifications result in a
range of complications to Eve: (i) increasing Eve's computational load
substantially, (ii) forcing her to do considerable online
computation rather than offline; and (iii) depriving her of any attack potential by  finally re-establishing ITS
for the modified construction.
% using a publicly known hash function composed with an SU$_2$ hash
% function.
We give necessary and sufficient conditions for ITS with
this construction; that the conditions are sufficient is already known
from earlier results, but that the conditions are necessary is, as far
as we know, a new result.

\subsection{Structure of the paper}
Section \ref{sec:AuthenticationInQKD}
contains a motivation on why authentication is
needed in QKD, shortly reviews message authentication codes and Universal hashing,
and gives a more detailed description of the
authentication method under study here.  Section \ref{sec:Attacks}
introduces the attack vectors and then details three different QKD
protocols and attacks against them in a step-by-step fashion.  In
Tables \ref{table:sifting-attacks} and
\ref{table:post-processing-attacks} we summarize the attacks and the
gained knowledge on the key for each of them, as well as for a number of further protocol
versions.  Section \ref{sec:Countermeasures} discusses how the
security of the authentication method can be improved and presents a
theorem that gives necessary and sufficient conditions for ITS
of the modified method.  The conclusions and outlook are given in Section
\ref{sec:Conclusions}.  The Appendices contain technical proofs and
summarize some definitions of Universal hash function families.

\section{Authentication in QKD}\label{sec:AuthenticationInQKD}

The need for authentication becomes clear if we consider for a moment
the opposite case, i.e.  an ``unprotected'' channel that allows
arbitrary modification of messages in transit.

\subsection{Man-in-the-middle attacks and Message Authentication Codes}\label{sec:AuthInQKD}
The unprotected channel
will enable a straightforward ``man-in-the-middle'' (MITM) attack:
%, also called ``intruder-in-the middle'' attack,
% based on impersonation:
\begin{definition}[straightforward man-in-the-middle (MITM) attack]\label{def:SF-MITM}
In the \emph{straightforward man-in-the-middle attack} the eavesdropper (Eve) builds or buys a
pair of QKD devices identical to those of the legitimate parties
(Alice and Bob) and cuts ``in the middle'' the quantum and classical
communication channels connecting Alice and Bob.  She now connects
each of her devices to the ``loose ends'' of the quantum and classical
channels and launches two \emph{independent} QKD sessions, one with
Alice and the other with Bob. Eve effectively pretends to be Bob to
Alice and Alice to Bob. Eventually she shares a (different) key with
each of the legitimate parties which allows her to communicate with
them independently. If Alice sends an encrypted message to Bob, Eve
can intercept the message and decrypt it, encrypt it with the key she
shares with Bob, and send it to Bob.
\end{definition}
Alice and Bob never come to
realize that the security of their communication is completely lost.
This is completely analogous to the classic MITM
attack against the unauthenticated Diffie-Hellman key agreement
protocol \cite[Chap.~12.9.1]{MenezesOV96}. Obviously, any (classic or
quantum) key agreement protocol that has no proper authentication
%access to identification of the parties \emph{and} has no access to message authentication can
(or integrity check) of messages exchanged between the communicating
parties can be broken with a similar impersonation attack.

So ideally an adversary should not be able to insert messages
into the channel, and moreover messages sent by one legitimate user to the
other are always delivered and are not modified.  However, there are no a-priori authentic
communication channels.
% In practice one settles for detecting inserted messages or changes in a
% transmitted message with high probability or at least put a high
% computational burden on the attacker.
Appending  a so-called Message
Authentication Code (MAC) to each communication message can mimic an authentic channel, but cannot
guarantee delivery of messages, as these can in practice always be
blocked.

\begin{definition}[Message Authentication Code (MAC)\cite{MenezesOV96}] A \emph{Message Authentication Code
    (MAC) algorithm} is a family of functions $h_K$ parameterized by a
  secret key $K$ with the following properties: (i) given a message
  $x$ and a key $K$, the MAC value $h_K(x)$ (also called \emph{tag})
  should be easy to compute, (ii) it maps a message $x$ of arbitrary
  finite bitlength to a tag $h_K(x)$ of fixed bitlength $n$, and (iii)
  given a description of the function family $h$, for every fixed
  allowable value of $K$ it must be computation-resistant. The last
  property means that given zero or more message-tag pairs
  $(x_i,h_K(x_i))$ it is computationally infeasible to compute any
  message-tag pair $(x,h_K(x))$ for any new input $x \ne x_i$ without
  knowing $K$.
\end{definition}

% \TODO{explain different classes of authenticity: authentic channel,
%   ``unprotected'' channel, epsilon secure channel: probability to
%   successfully insert/substitute any message by any other message is
%   bounded by epsilon, (should we mention delivery), authentic as in
%   our old paper: it is possible to substitute (no insertion!) any
%   given message by a set of messages}

Normally, MACs are either based on (a) cryptographic hash functions
(e.g. HMAC-SHA-256 based on SHA-256), on (b) block cipher algorithms
(e.g. AES-CMAC based on AES), or on (c) Universal$_2$ hashing (see
\ref{uhf}). Message authentication codes based on (a) or (b) typically
use one key for many messages, and offer computational security, i.e.
they can only be broken with sufficient computing power (or when a hidden
weakness of the algorithm is detected).

\subsection{Universal hashing and UC security}
MACs based on Universal$_2$
hashing have to use one (new) key per message, but offer
information-theoretic security which is independent of the adversary's
computing power.  In more detail, for SU$_2$ hash
functions, a random guess of the MAC tag is provably the best possible
attack, while $\epsilon$-ASU$_2$ hash functions
still provide a strict upper bound (namely $\epsilon$) on the
attacker's success probability to substitute an observed message-tag
pair with another valid message-tag pair (\emph{substitution attack})
or to insert a valid message-tag pair.

Universal hashing was originally proposed by Wegman and
Carter \cite{WegmanCarter81,CarterWegman79}.  It was identified
as an appropriate match for QKD, as Wegman-Carter's and later
constructions \cite{Krawczyk94,Stinson92,Mehlhorn84,Shoup96} consume
relatively low amount of key.  The aim is to have less key consumption
than the key generation in a typical QKD session \cite{BennettBBSS92},
so that each session can reserve a portion of its output for
authentication of the subsequent one. Then, the
process only needs to be kick-started by an initial, one-time,
pre-distributed secret.

Security analysis of QKD (see, e.g., Ref.~\cite{RMP} and references
therein for a recent overview) has typically been based on the
requirement that the classical post-processing communication is
secured by a MAC based on Universal hashing, to upper bound an
adversary's chances to modify or insert messages without getting
detected.
In addition UC-security definitions for QKD have been established
\cite{BenOr2004,BenOr2005,Renner_Koenig05,JMQRR08}. As a consequence
combining the two $\varepsilon$-UC-secure protocols QKD and ASU$_2$
authentication yields a joint, \emph{UC-secure key growing mechanism}
over non-authentic classical channels (see \cite{Portmann2012}).
Thus, MACs based on ASU$_2$ hashing are
sufficient for security, but it is an open question whether they are
also necessary, and what security would be obtained for other
alternatives.

\subsection{The non-ITS authentication mechanism of
  Ref.~\cite{Peev2005}}

The authentication mechanism proposed in
Ref.~\cite{Peev2005} aimed to
consume less key than ASU$_2$ authentication. The
intended goal is a positive key balance of the combination QKD plus
authentication even in realizations that use (relatively) short
blocks in the post processing step. Note that later experimental
progress has made these objectives not so relevant,
as short key blocks are no longer necessary from an implementation
perspective \cite{Sasaki}. Still, a
complete security analysis of the authentication mechanism of
\cite{Peev2005} is intriguing from a theoretical point of view as the
mechanism has interesting properties not shared by any of the methods
mentioned  above.

To start with, we summarize the proposal of Ref.~\cite{Peev2005}
and introduce some notation (see also Table~\ref{table:symbols}).
The proposal relies on a
two-step hash function evaluation: $t=g_K(m):=h_K(f(m))$, where
$f:\mM\to\mZ$ is a publicly known hash function and $h_K:\mZ\to\mT$
belongs to an SU$_2$ hash function family
$\mH$ (see \ref{uhf}).  Here, $\mM$ is the set of messages to be
authenticated, $\mZ$ is an intermediate set of strings, and $\mT$ is
the set of tags with $|\mM|\gg|\mZ|>|\mT|$.

%While we will study explicitely this two-step hash function,
%we note that a very similar analysis works if we replace in above construction
%the SU$_2$ family with an ASU$_2$ family.

%The SU$_2$ family
%$\mH=\{h_K:\mZ\to\mT\}$ has a minimum size $|\mZ|+|\mT|$, while an
%ASU$_2$ family $\{h:\mM\to\mT\}$ has a size proportional to
%$\log|\mM|$.  With the intended message lengths, the existing ASU$_2$
%constructions are larger than $\mH$.  Therefore, a shorter key $K$ is
%needed to identify the $h_K$ used in this construction than in ASU$_2$
%hashing.
\subsubsection{Insertion of messages is ruled out}
Now assume that Eve attempts to calculate or guess the tag for a fixed message
$\mE$ that she wants to insert. In that case she has a success probability of $1/|\mT|$ (irrespective of her
computing power).  This is because the key $K$ which identifies the
SU$_2$ hash function is not known to her.
Thus, the authentication mechanism is (first-)preimage resistant, i.e.\ knowledge of
the authentication tag alone does not allow to find messages yielding the same tag.
\subsubsection{Substitution with given messages is ruled out}
Let us further assume, Eve has intercepted a (valid) message-tag pair $(\mA,t)$ from Alice
and wants to substitute it with her \emph{fixed} message $\mE$ and some tag.
Then Eve's chances increase slightly
because she now has access to the intermediate value $f(\mA)$, and can
check if $f(\mA)=f(\mE)$. If there is a collision, Eve knows that
$(\mE,t)$ is a valid message-tag pair and can just send this,
otherwise she guesses the tag as above.  The total probability of
success is now bounded by the guessing probability plus the collision
probability, and assuming that $\mA$ is random to Eve and that $f$ is a
good hash function, the collision probability is low (for details see
\cite{Peev2005}).  So this two-step authentication works well in a situation when Eve is given a fixed
message $\mE$ to generate the tag for. One immediate consequence is that
Eve cannot perform the straightforward MITM attack
(cf.~Definition \ref{def:SF-MITM}) with significant success
probability since
she would need to generate tags for messages $\mE$ from her devices
without knowledge of $K$, for which case the above bound applies.
\subsubsection{The weakness}
However, one should note that using the intercepted
message-tag pair $(\mA,t)$ and enough computational power, Eve can in principle search for other preimages of $t$ under $f$.  If she can find (at least) one message $\tildemE$
such that $f(\mA)=f(\tildemE)$ then
$h_K(f(\mA))=h_K(f(\tildemE))$ and therefore $(\tildemE,t)$ is a valid message-tag pair
\emph{for any key $K$}. She can now replace $\mA$ with $\tildemE$ with success
probability of 100\%.  The question now is if this (one of these) $\tildemE$
can be used in place of the message $\mE$. It would seem that, if Eve
strictly follows the appropriate QKD protocol (random settings, best
possible bit error rate, \ldots), this is not possible.

However, Eve
is not forced to follow the precise requirements of the QKD protocol
\cite{Abidin2009}; she only needs to make it seem to Alice and Bob
that she does so. For example, Eve does not need to use random
settings (e.g.~preparation bases and raw keys), or even correctly send all settings
she used. If it helps her, she can use a fixed sequence of settings or
report other settings for some qubits than the ones actually used.

An early suggestion \cite{BMQ05} was to \emph{select} the privacy
amplification map carefully, rather than generating it randomly. This
would give Eve a shared key with Bob, but not with Alice. Later, as
mentioned above, it was observed that Eve may deviate from the QKD
protocol in several places \cite{Abidin2009}. If Eve uses a fixed sequence of
settings (e.g.~measurement and preparation bases) on the quantum channel
this would enable her to do the
calculations for finding $\tildemE$ offline.  If Eve sends the wrong settings for
some of the qubits this will allow her to choose from several $\tildemE$, to get
a collision. This would constitute the basis for a
% a full (but more sophisticated)
sophisticated MITM attack that can break simplified QKD protocols.  In these simplified protocols, the
breaches could be closed by relatively straightforward
countermeasures \cite{Peev2009}, but the security of the standard
and/or hardened protocols remained an open issue.
We aim to settle this in the present paper.

\section{Attacks against non-ITS authentication in QKD}\label{sec:Attacks}

In this section, we give detailed descriptions of four different
attacks on three different explicit QKD protocols. We also give an
overview of the effectiveness of this kind of attacks against  other QKD
protocols, and for different types of resources available to
Eve. In each case, the essence of the attack is to intercept a valid
message-tag pair (sent by Alice or Bob) and---using large computational resources
 and/or leveraging weaknesses of the public hash function algorithm---find 
further preimages of the tag (messages that hash to the same hash value as
the intercepted message) that are used by the eavesdropper.

\subsection{Probability for finding hash collisions in a set of messages}

\newcommand{\Prob}{\mathcal{P}}
\newcommand{\Pcoll}{\Prob^\textrm{succ}_\textrm{coll}}
\newcommand{\setM}{\mathcal{M}}

Assume that Eve has intercepted a message-tag pair $(\mA,t)$ from Alice.
The following lemma gives a lower bound for the probability that (under a mild
assumption) a set $\mathcal{M}$ of messages contains at least a single message $\mE$ that collides with $\mA$, i.e. 
$h_K(f(\mE))=t$.
\begin{lemma} \label{lem:collision1}
Let us assume that $f$ maps all messages in $\setM$ randomly onto $\mZ$.
%(for a randomly choosen $m\in\setB$ and $\forall z\in\mZ: \Pr\{f(m)=z\} = |\mZ|^{-1}$),
Then the probability that at least one of the messages in $\setM$ is validated by the given tag $t=h_K(f(\mA))$ is
\begin{equation*}
\Pcoll =\Pr\left\{\exists \mE\in\setM: h_K(f(\mE))=t) \right\} > 1-\exp\left(-|\setM| |\mZ|^{-1}\right).
\end{equation*}
\end{lemma}
The proof of Lemma \ref{lem:collision1} is given in \ref{ProofCollision1}. Since no assumptions on the computational power of Eve are imposed, she will be able to find with probability $\Pcoll$ such an $\mE$. Note, that $|\setM|=|\mZ|$ is sufficient to get $\Pcoll > 0.63$.

\subsection{Probability for finding a hash collision with small Hamming distance to a given message}\label{sec:HashCollisions}

Assume that Eve has intercepted a message-tag pair $(\mA,t)$ from Alice.
% In some cases, a colliding message needs to be close (in terms of Hamming
% distance) to a ``desired'' message that Eve wants a legitimate user
% to accept as authentic.  
The following corollary of Lemma~\ref{lem:collision1} states that (under a mild
assumption) for any fixed message $m^\Eve$, that Eve would like to
send, there exists with probability almost 1 a message
$\tilde{m}^\Eve$, such that (i) $\tilde{m}^\Eve$ is almost identical to $\mE$, 
i.e.\ $\tilde{m}^\Eve$ has \emph{small Hamming distance} to $m^\Eve$, and 
(ii) $(\tilde{m}^\Eve,t)$ will be accepted as authentic, i.e.\ $h_K(f(\tilde{m}^\Eve))=t$.
\newcommand{\setB}{\mathcal{B}}
\begin{corollary} \label{cor:collision2}
Let $\setB$ be the closed ball of all messages $m$ having a Hamming distance to $m^\Eve$ not exceeding $w$:
\begin{equation*}
\setB=\left\{ m: d_H(m,m^\Eve)\le w
\right\},
\end{equation*}
and let us assume that $f$ maps all messages in $\setB$ randomly onto $\mZ$.
%(for a randomly choosen $m\in\setB$ and $\forall z\in\mZ: \Pr\{f(m)=z\} = |\mZ|^{-1}$),
Then the probability that at least one of the messages in $\setB$ is validated by the given tag $t=h_K(f(\mA))$ is
\begin{equation*}
\Pcoll =\Pr\left\{\exists \tilde{m}^\Eve\in\setB: h_K(f(\tilde{m}^\Eve))=t) \right\} > 1-\exp\left(-|\setB| |\mZ|^{-1}\right).
\end{equation*}
For simplicity we can loosen the bound and replace $|\setB|$ by ${\ell \choose w} < |\setB|$, where $\ell$ is the length of the binary message $m^\Eve$.

% Under the assumption that a random message $m\in \mathcal{B}$
%chosen with uniform probability $|\mathcal{B}|^{-1}$
%has tag $t$ with probability $|\mZ|^{-1}|$, i.e.
%Assume that the collision probability of two distinct messages under $f$ does not depend on their Hamming distance (provided it is larger than some very small number).
\end{corollary}

The proof of Corollary~\ref{cor:collision2} is given in \ref{ProofCollision2}. Since no assumptions on the computational power of Eve are imposed, she will be able to find with probability $\Pcoll$ such an $\tilde{m}^\Eve$.
For typical parameters, e.g. $|\mZ|=2^{256}$, and $\ell=2^{12}$ ($2^{13}$, $2^{14}$, $2^{15}$, $2^{16}$, $2^{17}$), a Hamming distance $w=32$ ($28$, $25$, $22$, $20$, $19$) is sufficient to reach a success probability of 99.9\%.

\subsubsection{Attacking the sifting stage -- hiding in the noise}
Let us assume that during the sifting stage the legitimate parties will exchange messages that contain one bit per preparation/measurement basis (time slot).
%basis (e.g. 0 for the $+$ basis, and 1 for the $\times$ basis), resp.
Let us assume further that Eve can successfully attack the protocol (as discussed below), if she can substitute such a message, say $m^\Alice$, with a sifting message of her choice, say $m^\Eve$. From Corollary \ref{cor:collision2} it follows that if Eve replaces $m^\Alice$ with $\tilde{m}^\Eve$  instead of $m^\Eve$, she will introduce at this step (at most) an additional error $\epsilon=w/\ell\approx 0.78\%$ ($0.34\%$, $0.15\%$, $0.067\%$, $0.031\%$, $0.014\%$) %between Alice's and Bob's sifted key
(with parameters from above; in the worst case each modified basis bit could result in one flipped sifted key bit).
%That corresponds to a normalized (by the message length) Hamming distance of $w/l\approx 0.78\%$ ($0.34\%$, $0.15\%$, $0.067\%$, $0.031\%$, $0.014\%$).
This strategy allows Eve to hide the substitution of sifting messages in the usual noise on the quantum channel, since the following error correction step will also remove these small additional deviations. Obviously, the larger the message length $\ell$, the easier Eve's task is.

\subsubsection{Correlating the sifted keys of Alice and Bob}

Assume for the moment that Eve has intercepted the quantum bits from Alice and has saved them into her quantum memory. Assume further that she managed to fool Alice, so that Alice announces her the corresponding preparation bases. Then Eve can measure the quantum bits and get Alice's raw key.

The strongest of the presented attacks is based on the fact that once
Eve knows the raw key of Alice, she can by using a modification of
Bob's sifting message ensure with high probability that the complete
sifted key of Alice will be almost identical to that of Bob (cf.\ 
description of \hyperref[sec:prot1]{Protocol 1} and
\hyperref[hl:prot1:stepSe2s]{step (Se'') of the attack} against it.).
\newcommand{\Psift}{\mathcal{P}^\textrm{succ}_{\textrm{sift-attack}}}
\begin{lemma} \label{lem:sifting}
Let $\IA\in_R\{0,1\}^n$ be the raw key
%, and $\BA\in_R\{0,1\}^n$ be the string of corresponding preparation bases,
that Alice has used to prepare her quantum bits. Once Eve knows $\IA$
%and $\BA$
she can determine $\lfloor n/2 \rfloor - k$ bits of any fixed sifted key $\SE$ that she wants Alice to create with probability
\begin{equation}
\Psift\ge 1-\exp\left(-\frac{2k^2}{n}\right)
\end{equation}
by replacing Bob's sifting message with a message $\BAE$ that she has prepared.
\end{lemma}
%In the attack, Eve will be given two fixed bit sequences, (i) the sifted key she wants to achieve, $\SE$ and (ii) the raw key prepared by Alice, $\IA$.
Eve's attack will succeed if a \emph{subsequence} of $\SE$ (derived by
deleting some elements without changing the order) of length at least
$\lfloor n/2 \rfloor - k$ is also a \emph{subsequence} of $\IA$.  The
proof and a simple and efficient algorithm to generate $\BAE$ is given
in \ref{app:subsequence}. Note, that $k=O(\sqrt{n})$ is sufficient for
$\Psift\approx 1$.

%For example, as discussed in detail below,  sufficiently small modifications of the message Eve is supposed to send during sifting cause errors that, however,  do not increase the observed error rate on the channel perceivably.
%\begin{enumerate}
% \item In steps prior to error correction (i.e.\ quantum signal exchange and sifting) she performs certain MITM attack and finds hash collisions between all messages $m_i^A$, $m_j^\Bob$ that Alice and Bob send, and corresponding messages $\tilde{m}_i^{EA}$, $\tilde{m}_j^{EB}$ that are ``close'' to messages $m_i^{EA}$, $m_j^{EB}$, that she would send if there were no authentication at all. Here ``close'' means that the relative Hamming distance between $m$ and $\hat{m}$, i.e.\ $\textrm{weight}(m\oplus \hat{m})/\textrm{length(m)}$, should be sufficiently small, so that errors would not increase the observed error rate on the channel perceivably.
% \item In the remaining phases During error correction Eve listens to the syndrom and corrects her sifted key. After that she listens to Alice and performs as Bob does.
%\end{enumerate}
%
%
% the Hamming distance of which
%    to $m^\prime$ is small (i.e.\ $\textrm{weight}(m\oplus \hat{m})/\textrm{length(m)}$, is small)
%
%% I denotes immediate authentication, i.e.\ each message gets its own MAC. P denotes postposed authentication, i.e.\ after sending a sequence of messages, a single MAC is used to authenticate the complete sequence.

\newcommand{\fake}{\textrm{fake}}

\subsection{General remarks, protocol notation and settings used}
\label{sec:intro-attacks}

Any successful attack is based on finding protocol modifications yielding communication messages that collide with  those of the legitimate parties under the fixed hash function  in the first (internal) stage of authentication  throughout the complete chain of the QKD protocol. Therefore, in contrast to the case of authentication by universal hashing, now QKD post-processing protocols  differing in the precise definition of their separate algorithmic steps (e.g. mode of authentication --- immediate or delayed, exact order of exchange of sifting messages,
whether error-correction bits are encrypted or not,
etc.) become inequivalent and exhibit different types of vulnerabilities. For this reason each attack discussed below is adapted to a specific protocol. Both the protocols and the corresponding attacks are carefully and formally defined.

We consider exclusively but without loss of generality  the case of BB84 QKD protocols, as the attacks we discuss are essentially independent of the particular form of quantum communication. Moreover,  all protocols that we study are stated as prepare-and-measure ones. It is, however, straightforward to adapt  the attacks discussed below to the case of entanglement based protocols.

It is implicitly assumed that on receiving messages Alice and Bob
check their message tags for correctness, and that incorrect message
tags lead them to conclude that Eve is intercepting, and to abort the
protocol. In case the message authentication is UC-secure the 
resulting protocols are also UC-secure.
A collection of used symbols is given in
Table~\ref{table:symbols}.

\begin{table}[tp]
\caption{Summary of symbols used in the paper.\label{table:symbols}}
\begin{tabularx}{\textwidth}{lX}
\br
Symbol & Description \\
\mr
\Alice, \Bob, \Eve & Legitimate parties: Alice, Bob; and eavesdropper Eve. \\
$\qchan$, $\cchan$ & quantum channel, classical channel \\
$\BA$ ($\BE$), $\IA$ ($\IE$) & Alice's (Eve's) string for bases choice and raw key, resp., used for preparing the quantum states. \\
$\BB$, $\IB$ & Bob's bases choice and measurement results (i.e.\ his raw key). \\
$\sA$ ($\sE$) & quantum state, prepared by Alice (Eve). \\
$\mack$ & notification that a party has finished its measurements. \\
$\MAC{K}{\cdot}$ & keyed hash function with key $K$. \\ %($K$ is taken from the key store). Note: The key store is not a generally known concept. This is a SECOQC and AIT "slang".
$\Bases^{\common{\textrm{X}}{\textrm{Y}}}$ & string indicating the positions where the parties X and Y successfully prepared and measured in the same basis. \\ %, i.e.\ $\Bases^{\textrm{X}}_i =\Bases^{\textrm{Y}}_i \iff \Bases^{\common{\textrm{X}}{\textrm{Y}}}_i = 1, \Bases^{\textrm{X}}_i \ne \Bases^{\textrm{Y}}_i \iff \Bases^{\common{\textrm{X}}{\textrm{Y}}}_i = 0 $.
$\SA$ ($\SB,\SE$) & sifted key of Alice (Bob, Eve). \\
$\SEA$ ($\SEB)$ & sifted key shared between Eve and Alice (Bob). \\
$\hatSB$ & error corrected (reconciled) key of Bob. \\
$\hatSEA$ & error corrected (reconciled) key that Eve shares with Alice. \\
$\KA$ ($\KB,\KE$) & final key of Alice (Bob, Eve). \\
$\KEA$ ($\KEB$) & final key shared between Eve and Alice (Bob). \\
$\ECCode:=\{\ECCode_1,\dots,\ECCode_n\}$ & set of predefined parity check matrices, used for forward error correction in different error rate regimes. \\
$i$ & index into the set $\ECCode$, denoting the actual parity check matrix $\ECCode_i$ used. \\
$\COCode$ & description of (ITS) confirmation function. \\
$\PACode$ & description of (ITS) privacy amplification function. \\
$\errorr$ & error rate on $\qchan$. \\
$\fail$ & notification that a partner should abort protocol. \\
\br
% \item $\plus=\{\rightarrow,\uparrow\}$, $\cross=\{\searrow,\nearrow\}$: state bases.
% \item [state encoding] In basis $\plus$ the bit value 0 (1) is encoded with $\rightarrow$ ($\uparrow$), in basis $\cross$ the bit value 0 (1) is encoded with $\uparrow$  ($\nearrow$).
% \item $\IA \in_r \{0,1\}^N$, and $\BA\in_r\{\plus,\cross\}^N$ denote Alice's random strings for data bits and bases choice.
% $\IB \in \{0,1,\myempty\}^N$ denotes the string of measurement results\footnote{$\IB_k=\myempty$ iff Bob's detectors did not produce a measurement result (e.g. due to absorption in the channel, or the imperfection of the detectors).}.\\
% $\state^{\Alice,\Eve} \in_r \{\uparrow,\rightarrow,\nearrow,\searrow\}^N$ denotes the prepared quantum state.\\
\end{tabularx}
\end{table}

%\Alice\ and \Bob\ must provide time synchronization between the state preparation and the measurement.
%\subsection{Protocol 1 -- BB84 with immediate message authentication}\hypertarget{hl:sec:prot1}{}
\subsection{Protocol 1 -- BB84 with immediate message authentication
  -- Alice sends bases}
\label{sec:prot1}

We divide the protocol into two separate parts:
\hyperref[sec:prot1-qsift]{(S) quantum state transmission and sifting},
and \hyperref[sec:prot1-postproc]{(P) post processing} (consisting of error
correction, confirmation, and privacy amplification). Part (P) needs
the result of (S) (i.e.\ the sifted keys) as input.

\begin{figure}
  \centering
  \includegraphics[width=\textwidth]{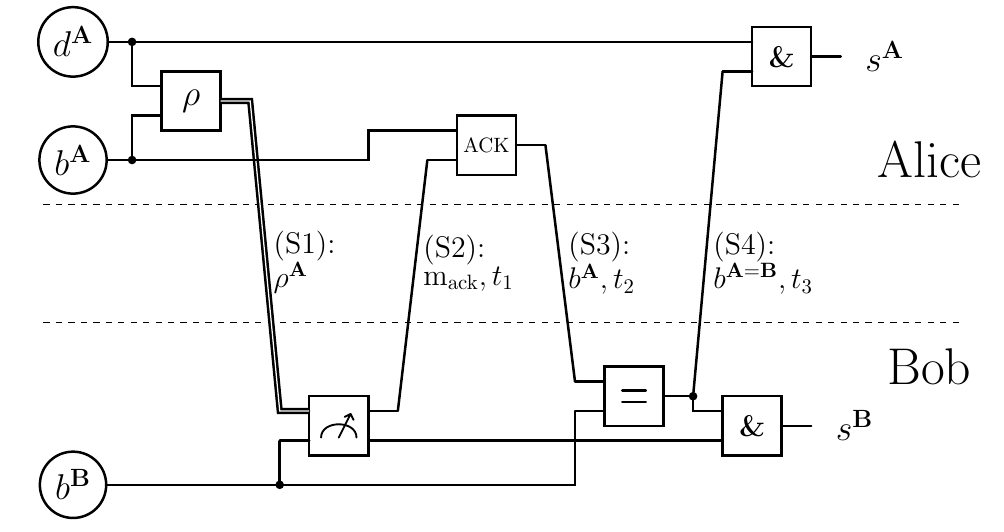}
  \caption{Protocol 1 (BB84, Quantum exchange and sifting only). Time flow is from left to right. Single (double) lines represent classical (quantum) communication. Local protocol actions are depicted by boxes: $\rho$ depicts state preparation, the indicator is a quantum measurement device, the ACK box denotes that Alice waits for Bob's message until she continues with the protocol, = denotes the calculation of identical bases, \& denotes the filtering of signals (in different bases).}
  \label{fig:attack1}
\end{figure}

\subsubsection{State transmission and sifting (S)}\label{sec:prot1-qsift} \ \\
SUMMARY: 3 classical messages are exchanged. Each classical message is accompanied by a corresponding tag (keyed hash value, MAC).
% RESULT: Alice and Bob share correlated bit strings (sifted keys) $\SA$, $\SB$.
\begin{enumerate}[1.]
 \item \Setup $\Alice$ and \Bob\ share the $3$  keys $K_1,K_2,K_3$.
 \item \Messages Let $t_1:=\MAC{K_1}{\mack}$, $t_2:=\MAC{K_2}{\BA}$, and $t_3:=\MAC{K_3}{\BAB}$ be the authentication tags used in messages (S2), (S3), and (S4), resp.
   \begin{enumerate}[(S1)]
     % This (black) magic fixes the reference numbers to the messages
     \makeatletter
     \renewcommand{\p@enumii}{S}
     \makeatother
   \item $\Alice \stackrel{\qchan}{\longrightarrow} \Bob:\quad \sA$ \label{prot1:s1} % \tag{S1}
   \item $\Alice \stackrel{\cchan}{\longleftarrow} \Bob:\quad \mack,t_1$ \label{prot1:s2} % \tag{S2}
   \item $\Alice \stackrel{\cchan}{\longrightarrow} \Bob:\quad \BA,t_2$ \label{prot1:s3} % \tag{S3}
   \item $\Alice \stackrel{\cchan}{\longleftarrow} \Bob:\quad \BAB,t_3$ \label{prot1:s4} % \tag{S4}
   \end{enumerate}
%\item \Bob\ sends ``measurement finished'' (with an additional timestamp\footnote{The timestamp rules out ``replay'' attacks against authentication}) on the classical channel to Alice.

\item \Actions
\begin{enumerate}[(Sa)]
\item \hypertarget{prot1:s:q}{} \Alice\ creates two random bit strings, her raw key $\IA$, and the bases string $\BA$, $\IA, \BA \in_r \{0,1\}^N$. For all pairs of bits $\left(\IA_k,\BA_k\right)$ \Alice\ generates the corresponding quantum states $\sA_k\in \{\state^0,\state^1,\state^2,\state^3\}$.
Using $\qchan$, \Alice\ sends the quantum state $\sA=\bigotimes_{k=1}^N\sA_k$ (``string'' of all $\sA_k$'s), i.e.\ \eqref{prot1:s1} to \Bob.

\item \hypertarget{prot1:s:Bm}{} \Bob\ creates a random bases string
  $\BB\in_r\{0,1\}^N$. \Bob\ measures $\sA$ in bases $\BB$ and obtains
  $\IB \in \{0,1,\myempty\}^N$, where $\myempty$ corresponds to no
  measurement result at \Bob, e.g., due to absorption in the channel,
  or the imperfection of the detectors.  For all $k$ with
  $\IB_k=\myempty$, \Bob\ sets $\BB_k=\myempty$.

\item \hypertarget{prot1:s:ack}{} Using $\cchan$, \Bob\ sends an acknowledgement message \eqref{prot1:s2} to \Alice.

\item \Alice\ waits until she has received \eqref{prot1:s2}, ensuring
  that the measurements have been finished before bases exchange is
  performed.  Using $\cchan$, \Alice\ sends \eqref{prot1:s3} to \Bob.

\item \Bob\ calculates a bit string $\BAB$, such that $\BAB_k=1$, if $\BA_k=\BB_k$, and $\BAB_k=0$, otherwise, for $1\le k \le N$.
% ($\BAB_k=1$ if \Alice\ and \Bob\ prepared and measured $\sA_k$ in the same basis, and $\BAB_k=0$ if \Bob\ did not measure a signal or the bases were different).
\Bob\ removes from $\IB$ all bits $\IB_k$ where $\BAB_k=0$ and obtains $\SB$.
Using $\cchan$, \Bob\ sends \eqref{prot1:s4} to \Alice.

\item \Alice\ removes from $\IA$ all bits $\IA_k$ where $\BAB_k=0$ and obtains $\SA$.
\end{enumerate}
\end{enumerate}

\subsubsection{Post processing (P)
%Error correction, confirmation, privacy amplification
}\label{sec:prot1-postproc} \ \\
SUMMARY: 3 classical messages with MACs are exchanged.
% RESULT: (Assuming the authentication is not broken) Alice and Bob share an ITS key $\KA=\KB$ (which can be of length zero).
\begin{enumerate}[1.]
 \item \Setup $\Alice$ and \Bob\ share $3$ keys $K_4,K_5,K_6$.
 \item \Messages Let $\TA= (i,\ECCode_i(\SA),\COCode,\COCode(\SA))$.
   \begin{enumerate}[(P1)]
     % This (black) magic fixes the reference numbers to the messages
     \makeatletter
     \renewcommand{\p@enumii}{P}
     \makeatother
   \item $\Alice \stackrel{\cchan}{\longrightarrow} \Bob: \quad \msg[K_4]{\TA}$ \label{prot1:p1} % (P1)
   \item $\Alice \stackrel{\cchan}{\longleftarrow} \Bob: \quad \msg[K_5]{\errorr} \quad / \quad \msg[K_5]{\fail}$ \label{prot1:p2} %(P2)
   \item $\Alice \stackrel{\cchan}{\longrightarrow} \Bob: \quad \msg[K_6]{\PA} \quad / \quad \text{-----}$ \label{prot1:p3} % (P3)
   \end{enumerate}
 \item \Actions
 \begin{enumerate}[({P}a)]
   % This (black) magic fixes the reference numbers to the steps
   \makeatletter
   \renewcommand{\p@enumii}{P}
   \makeatother
\item \Alice\ estimates the parameters of $\qchan$ (based on the error rate of previous rounds or by choosing a default value), selects a corresponding forward error correction algorithm $\ECCode_i$ from a predefined set, and calculates the syndrome $\ECCode_i(\SA)$. \Alice\ determines a confirmation function $\COCode$, and calculates $\COCode(\SA)$. $\Alice$ sends \eqref{prot1:p1}.\label{prot1:pa}
\item \Bob\ uses $\ECCode_i$ and $\ECCode_i(\SA)$ to correct $\SB$ resulting in $\hatSB$. \Bob\ uses $\COCode$ to calculate $\COCode(\hatSB)$. \Bob\ checks whether $\COCode(\hatSB)=\COCode(\SA)$. If the identity holds, \Bob\ calculates the error rate
%$\epsilon=|\{k:\SA_k\ne\SB_k\}|/\length \SA$
$\epsilon$ and sends it to $\Alice$ \eqref{prot1:p2}. If not, \Bob\ sends \emph{fail} to $\Alice$ \eqref{prot1:p2} and aborts the protocol.\label{prot1:pb}
\item If $\Alice$ receives $\epsilon$, $\Alice$ determines a corresponding privacy amplification function $\PA$, calculates $\KA=\PA(\SA)$, and sends \eqref{prot1:p3}. If $\Alice$ receives \emph{fail} she aborts the protocol.\label{prot1:pc}
\item If \Bob\ has not aborted in step \eqref{prot1:pb}, he now calculates $\KB=\PA(\hatSB)$.
With probability almost 1 (determined by the confirmation function CO), $\KA=\KB$.\label{prot1:pd}
\end{enumerate}

\end{enumerate}

\subsubsection{Attack against Protocol 1}\label{sec:prot1-attack} \ \\
\noindent Eve replaces the quantum channel between Alice and Bob with
ideal quantum channels and her instrumentation to prepare, store, and
(almost) perfectly measure quantum states.

%SUMMARY: Eve uses a quantum memory, prepares quantum states, and modifies sifting messages.\\
\noindent RESULT: Alice, Bob, and Eve share identical keys $\KA=\KB=\KE$.
\begin{enumerate}[1.]
 \item \Notation \\
 $\tildeBases^x$: a string that deviates slightly from $\Bases^x$ to reach a hash collision with a given tag $t$ [used in messages (S3') and (S4')].
% \item System setup: $\Alice$ and \Bob\ must provide time synchronization between the state preparation and the measurement
 \item \Attackmessages Let $t_1:= \MAC{K_1}{\mack}$, $t_2:=\MAC{K_2}{\BA}$, and $t_3:=\MAC{K_3}{\BEB}$ be the authentication tags used in messages (S2)--(S4).

\setlength{\jot}{-0.5mm}
%\begin{tabular}{cp{0.03\textwidth}cp{0.03\textwidth}lp{0.05\textwidth}p{0.3\textwidth}l}
%\Alice & $\stackrel{\qchan}{\longrightarrow}$ & \Eve & & & : & $\sA$ & (S1) \\
%& & \Eve & $\stackrel{\qchan}{\longrightarrow}$ & \Bob & : & $\sE$ & (S1') \\
%\Alice & & $\stackrel{\cchan}{\longrightarrow}$ & & \Bob & : & $\msg[K_4]{\TA}$ & (S5) \\
%%A \stackrel{\cchan}{\longleftarrow} \Bob: & \msg[K_5]{\errorr} \quad / \quad \msg[K_5]{\fail} \qquad & (6) \\
%%A \stackrel{\cchan}{\longrightarrow} \Bob: & \msg[K_6]{\PA} \quad / \quad \text{-----} & (7) \\
%  \end{tabular}
%
%\textit{ \textbf{Comment}: The three lines above are not valid text. They are intentionally left to show a different possible format for the attack representation}

\begin{enumerate}[(S1)]
\item $\Alice \stackrel{\qchan}{\longrightarrow} \Eve: \quad \sA$ % (S1)
\item[(S1')] $\Eve \stackrel{\qchan}{\longrightarrow} \Bob: \quad \sE$ % (S1')
\item $\Alice \stackrel{\cchan}{\longleftarrow} \Bob: \quad \mack,t_1$ % (S2)
\item $\Alice \stackrel{\cchan}{\longrightarrow} \Eve: \quad \BA, t_2$ % (S3)
\item[(S3')] $\Eve \stackrel{\cchan}{\longrightarrow} \Bob: \quad \tildeBE, t_2$ % (S3')
\item $\Eve \stackrel{\cchan}{\longleftarrow} \Bob: \quad \BEB, t_3$ % (S4)
\item[(S4')] $\Alice \stackrel{\cchan}{\longleftarrow} \Eve: \quad \tildeBAE, t_3$ % (S4')
\item[(P1)] $\Alice \stackrel{\cchan}{\longrightarrow} \Bob: \quad \msg[K_4]{\TA}$ % (P1)
\item[(P2)] $\Alice \stackrel{\cchan}{\longleftarrow} \Bob: \quad \msg[K_5]{\errorr} \quad / \quad \msg[K_5]{\fail}$ % (P2)
\item[(P3)] $\Alice \stackrel{\cchan}{\longrightarrow} \Bob: \quad \msg[K_6]{\PA} \quad / \quad \text{-----}$ % (P3)
  \end{enumerate}

 \item \Attackactions %($\Alice$ and \Bob\ adhere to the protocol of Sec.~\ref{sec:prot1-qsift}. We only state the actions performed by \Eve.
\begin{enumerate}[(Sa)]
\item \Alice\ performs step (Sa) of the protocol (prepares $\sA$ and sends it in (S1)).
\item[(Sa')] \Eve\ intercepts (S1) from $\Alice$ and stores $\sA$ in
  her quantum memory. Then \Eve\ performs exactly as \Alice\ in step
  (a) of the protocol:  \Eve\ determines random $\IE$ and $\BE$,
  prepares a state $\sE$ and sends it in (S1') to \Bob. 
\item \Bob\ performs step (Sb) of the protocol measuring the state
  \Eve\ has prepared, $\sE$, instead of $\sA$, as in the protocol (in
  the following denoted as $\sA \rightarrow \sE$). 
\item \Bob\ performs step (Sc) of the protocol, i.e.\ he sends (S2).
\item \Alice\ performs step (Sd) of the protocol. She sends (S3).
\item[(Sd')] \Eve\ intercepts (S3), i.e.\ $\BA$ and the corresponding tag $t_2$, and measures her quantum memory in bases $\BA$ and obtains an identical copy of \Alice's raw key, $\IA$. 
\item[(Sd'')] \Eve\ determines $\tildeBE$ (e.g.~using an exhaustive
  search), such that the intercepted $t_2$ validates $\tildeBE$ and
  $d_H(\tildeBE,\BE)$ is small (cf.~Corollary \ref{cor:collision2}), and
  sends (S3') to \Bob.
\item \Bob\ performs step (Se) of the protocol ($\BA \rightarrow
  \tildeBE$, $\BAB \rightarrow \BEB$), obtains $\SB$ and sends message
  (S4).
\item[(Se')] \Eve\ intercepts (S4), i.e.\ $\BEB$ and the corresponding
  tag $t_3$. \Eve\ removes from $\IE$ all bits $\IE_k$ where
  $\BEB_k=0$ and obtains $\SEB\approx \SB$ (in general $\SEB\ne\SB$
  because \Eve\ had to send $\tildeBE$ instead of her true basis
  choice $\BE$ in step (Sd'')).
\item[(Se'')] \label{hl:prot1:stepSe2s} Using the algorithm
  detailed in \ref{alg:bases1}, \Eve\ searches for a subsequence of $\IA$ that coincides with $\SEB$ and  calculates $\BAE$ such that in
  \Alice's next step, (Sf), \Alice\ would create $\SA\approx \SEB$ as
  her sifted key. Typically \Eve\ will have to allow for $O(\sqrt{n})$
  bits that will be different between $\SA=\SEA$ and $\SEB$ (see Lemma
  \ref{lem:sifting}).
\item[(Se''$\!$')] As in step (Sd'') \Eve\ determines $\tildeBAE$ with
  small Hamming distance to $\BAE$, this time validated by $t_3$
  obtained in step (Se'), calculates the actual sifted key of \Alice,
  $\SEA\approx\SEB$ and sends (S4') to \Alice.
\item \Alice\ performs step (Sf) of the protocol ($\BAB \rightarrow
  \tildeBAE$) and obtains $\SA=\SEA$.
\end{enumerate}

Note: Eve has almost reached her goal, as $\SA=\SEA\approx\SEB\approx\SB$ holds. The subsequent error correction step allows her to reach $\KA=\KE=\KB$:

\begin{enumerate}[(Pa)]
\item $\Alice$ performs step (Pa) of the protocol. Eve reads (P1), and
  uses the syndrome to correct her sifted key (in case \Alice's
  preparation and/or \Eve's quantum measurement and preparation are not
  100\% perfect, so that $\SEA\approx\SA$).
\item \Bob\ performs step (Pb) of the protocol: $\SA=\SEA=\hatSB$.
\item $\Alice$ performs step (Pc) of the protocol and obtains $\KA=\PA(\SA)$. 
\item[(Pc')] \Eve\ reads (P3),  the privacy amplification function $\PA$. \Eve\ calculates $\KE=\PA(\SEA)=\KA$.
\item \Bob\ performs step (Pd) of the protocol: $\KA=\KE=\KB$.
\end{enumerate}
\end{enumerate}

This attack completely breaks protocol 1. Eve has an identical copy of Alice's and Bob's shared ``secret'' key. This is the strongest possible attack. For instance, using her copy of the key, Eve can simply decrypt messages from, and encrypt and/or authenticate new messages to both parties.

If this key is used to authenticate further QKD rounds, Eve can now continue with a much simpler impersonation attack, in which she does not have to calculate hash collisions or use her quantum memory.

%Note: in step (vii) Eve could send any slot information and consequently generate an identical key between Eve and Alice, i.e.\ $K_A=K_{E1}$, but would be different from Bob's key $K_B=K_{E2}$.
%

\begin{figure}
  \centering
  \includegraphics[width=\textwidth]{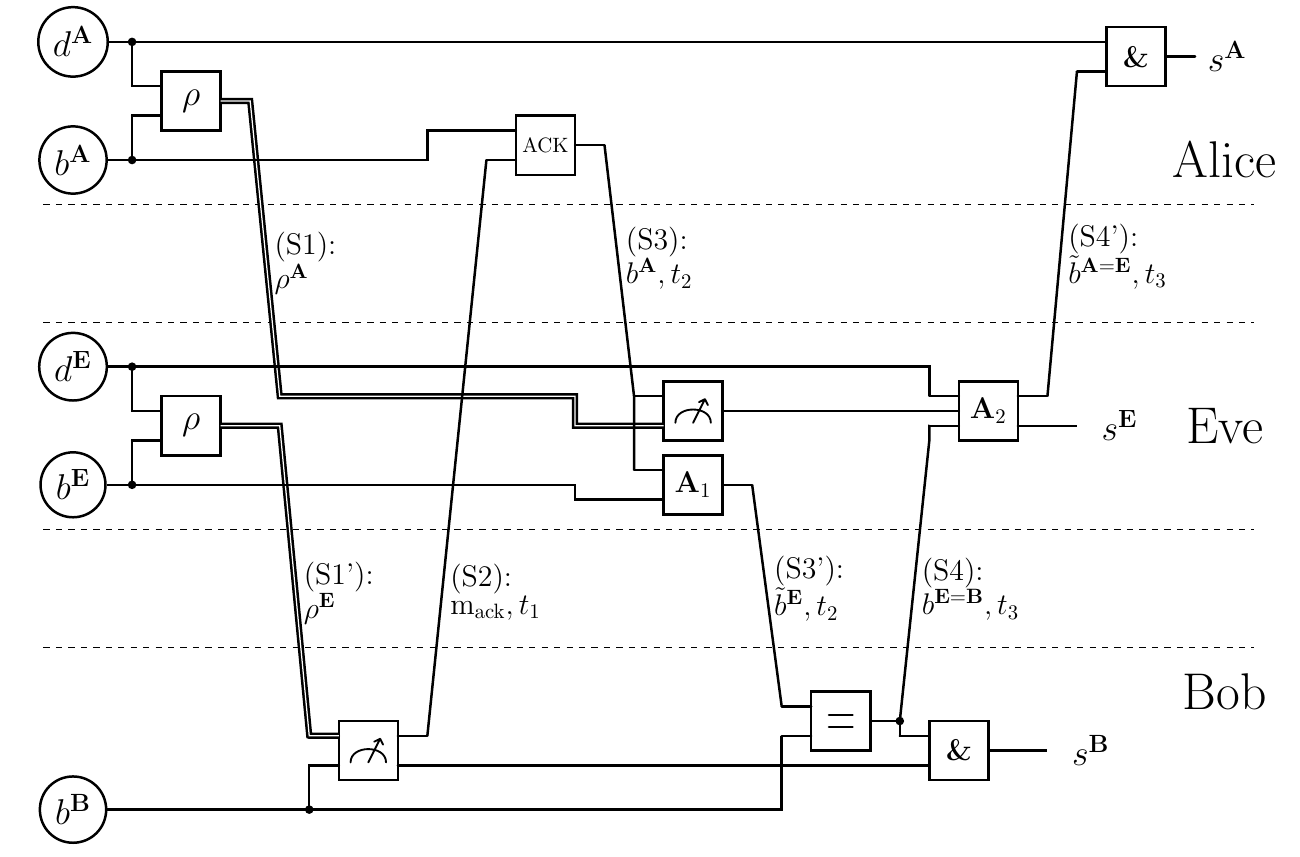}
  \caption{Interleaving attack against quantum exchange and sifting of
    \hyperref[sec:prot1]{Protocol 1},
    % \hyperlink{hl:sec:prot1}{Protocol 1}, %this does not
    % compile due to a bug with my caption version 3.1k, it seems
    % to be fixed in 3.1l:
    % http://tex.stackexchange.com/questions/3938/problem-with-caption-and-hyperref
    a QKD-protocol with immediate authentication. Time flow is from
    left to right. Single (double) lines represent classical (quantum)
    communication. See caption of Fig.~\ref{fig:attack1} for a
    description of boxes and symbols. The new boxes A$_{1,2}$ denote
    the attack actions, described in protocol steps (Se) through
    (Se''$\!$'). Employing quantum memory Eve manages to bring Alice
    and Bob to distill a sifted key that she knows with probability
    approaching 1.}
\end{figure}

\subsection{Protocol 2 -- BB84 with delayed message authentication -- Alice sends bases}\label{sec:prot2}
This protocol is very similar to \hyperref[sec:prot1]{Protocol 1}, the
difference is the authentication method: the authentication is delayed
and performed only at the end of the protocol verifying the integrity
of all messages. This, however, will change details of our attack
strategy: until the very last message we don't have to care about
authentication, but at the end we attack the privacy amplification
matrix to get enough degrees of freedom to find collisions (step
(Pc'), see below).

%\subsubsection{Protocol steps}\label{sec:prot2-qsift} \ \\
\noindent SUMMARY: 7 classical messages are exchanged. A nonce is used to enforce synchronization. The two last messages are authenticated with MACs.
% RESULT: Alice and Bob share an ITS key $\KA=\KB$ (which can be of length zero).
\begin{enumerate}[1.]
 \item \Notation $\nB$: random number (nonce), created by \Bob.
 \item \Setup $\Alice$ and \Bob\ share two keys $K_1,K_2$.
 \item \Messages Let $\TA= (i,\ECCode_i(\SA),\COCode,\COCode(\SA))$, $M^\Alice=(\BA,\TA,\PA)$, $M^\Bob=(\nB,\BAB,\errorr/\fail)$.
   \begin{enumerate}[(S1)]
   \item $\Alice \stackrel{\qchan}{\longrightarrow} \Bob: \quad \sA$ % (S1)
   \item $\Alice \stackrel{\cchan}{\longleftarrow} \Bob: \quad \nB$ % (S2)
   \item $\Alice \stackrel{\cchan}{\longrightarrow} \Bob: \quad \BA$ % (S3)
   \item $\Alice \stackrel{\cchan}{\longleftarrow} \Bob: \quad \BAB$ % (S4)
   \item[(P1)] $\Alice \stackrel{\cchan}{\longrightarrow} \Bob: \quad \TA$ % (P1)
   \item[(P2)] $\Alice \stackrel{\cchan}{\longleftarrow} \Bob: \quad {\errorr}\ /\ \msg{\fail},\MAC{K_1}{M^\Bob}$ % (P2)
   \item[(P3)] $\Alice \stackrel{\cchan}{\longrightarrow} \Bob: \quad \msg{\PA},\MAC{K_2}{M^\Alice}\ / \ \text{-----}$ % (P3)
  \end{enumerate}
\item \Actions Steps (Sa)--(Sf) and (Pa)--(Pd) are identical to that
  of protocol 1, with the following exceptions: (a) only the two last
  messages of the protocol, (P2) and (P3), [which are sent in step
  (Pb) and (Pc)] have MACs attached that authenticate all messages
  from Bob to Alice and Alice to Bob, respectively, (b) in step (Sc)
  the message (S2) contains a nonce $\nB$, a random number that is
  chosen by Bob and used to ensure that Bob has
  finished measuring before the bases exchange starts. Using a fixed
  $\mack$ as in protocol 1 instead of the random nonce $n^\Bob$ would
  allow for a trivial attack. 
  % Eve could (i) store $\sA$ in her quantum memory, (ii) ``pre-play''
  % $\mack$ to trigger Alice's message (S3) containing her preparation
  % bases $\BA$, and (iii) measure her quantum memory in Alice's
  % preparation bases to get complete information on $\IA$. Since
  % Alice authenticates $\nB$ in message (P3), Eve cannot replace it
  % with a number of her choice.
\end{enumerate}

\subsubsection{Attack against Protocol 2 (Eve only attacks messages to Bob)}\label{sec:prot2-attack}

Eve replaces the quantum channel between Alice and Bob, with ideal quantum channels and her instrumentation to prepare, store and perfectly measure quantum states.
The first part of the attack is similar to the attack against protocol 1 but it differs in several essential instances. All steps from (Sa) to (Sd')  are basically the same, but messages (S2) and (S3) are sent without MACs. From now on the attack differs so that Eve can cope with the form of postponed authentication utilized in protocol 2. In particular, we assume that Eve cannot manipulate the message that contains the error rate $\errorr$ on the quantum channel. This could be the case, for example, if $\errorr$ is encoded as 16 bit integer: the existence of hash collisions is very unlikely, since it is impossible to reach the needed Hamming distance of at least 19 (see Sec.~\ref{sec:HashCollisions}). This in turn implies, that Eve can also not manipulate any previous message from Bob to Alice (since she does not know what value of $\errorr$ Bob will be transmitting, she does not know which messages to prepare to get a hash collision). In particular, Eve cannot modify the sifting message of Bob, which rules out an attack analogous to the attack against protocol 1, described above. Amazingly, although Eve cannot modify any message from Bob to Alice, she can still mount the most powerful attack (Alice, Bob, and Eve share the same key)!

%SUMMARY: Eve uses a quantum memory, prepares quantum states, and modifies sifting messages.\\
\noindent RESULT: Alice, Bob, and Eve share identical keys $\KA=\KB=\KE$.

\begin{enumerate}[1.]
\item \Attackmessages In addition to the definitions in the protocol above, let $t_2=\MAC{K_2}{M^\Alice}$.
  \begin{enumerate}
  \item[(S1)] $\Alice \stackrel{\qchan}{\longrightarrow} \Eve: \quad \sA$
  \item[(S1')] $\Eve \stackrel{\qchan}{\longrightarrow} \Bob: \quad \sE$
  \item[(S2)] $\Alice \stackrel{\cchan}{\longleftarrow} \Bob: \quad n^\Bob$
  \item[(S3)] $\Alice \stackrel{\cchan}{\longrightarrow} \Eve: \quad \BA$
  \item[(S3')] $\Eve \stackrel{\cchan}{\longrightarrow} \Bob: \quad \BE$
  \item[(S4)] $\Alice \stackrel{\cchan}{\longleftarrow} \Bob: \quad \BEB$
  \item[(P1)] $\Alice \stackrel{\cchan}{\longrightarrow} \Eve: \quad \TA$
  \item[(P1')] $\Eve \stackrel{\cchan}{\longrightarrow} \Bob: \quad \TE$
  \item[(P2)] $\Alice \stackrel{\cchan}{\longleftarrow} \Bob: \quad {\errorr}\ /\ \msg{\fail},\MAC{K_1}{M^\Bob}$
  \item[(P3)] $\Alice \stackrel{\cchan}{\longrightarrow} \Eve: \quad \msg{\PA},t_2\ / \ \text{-----}$
  \item[(P3')] $\Eve \stackrel{\cchan}{\longrightarrow} \Bob: \quad \msg{\PE},t_2\ / \ \text{-----}$
  \end{enumerate}
 \item \Attackactions 
   \begin{enumerate}
   \item[(Sa)] -- (Sd') Identical to those of protocol 1 (cf.\
     Sec.~\ref{sec:prot1-attack}), up to the absence of authentication
     tags in the present protocol.
   \item[(Sd'')] \Eve\ performs step (Sd) of the protocol ($\BA
     \rightarrow \BE$) and sends message (S3') to \Bob.
   \item[(Se)] \Bob\ performs step (Se) of the protocol ($\BA
     \rightarrow \BE$), obtains $\BEB$ and $\SB$, and sends message
     (S4).
   \item[(Se')] \Eve\ reads message (S4), i.e.\ $\BEB$. She removes
     from $\IE$ all bits $\IE_k$\ for $k: \BEB_k=0$ and obtains $\SEB
     =\SB$, possibly with noise.
   \item[(Sf)] \Alice\ performs step (Sf) of the protocol ($\BAB
     \rightarrow \BEB$) and obtains $\SA$.
   \item[(Sf')] \Eve\ removes from the string $\IA$ (which she knows
     exactly) all bits $\IE_k$\ for $k:\BEB_k=0$ and obtains
     $\SEA=\SA$.
\end{enumerate}
Note: Eve now shares two keys with Alice and Bob respectively $\SA=\SEA$ and $\SEB=\SB$ (or $\SEB\approx\SB$ as discussed above) but these keys are not correlated. After the subsequent error correction step \Eve\ already shares $\hatSA=\hatSEA$ and $\hatSEB=\hatSB$. Finally, attacking the privacy amplification step of the protocol \Eve\ succeeds in achieving her ultimate goal $\KA=\KE=\KB$:

\begin{enumerate}[(Pa)]
\item \Alice\ performs this step in the protocol and sends message (P1).
\item[(Pa')] \Eve\ intercepts (P1), produces
  $\TE=(i,\ECCode_i(\SEB),\COCode,\COCode(\SEB))$ and sends message
  (P1') to \Bob. (If \Eve\ would anticipate an error between her and
  \Bob\ that is too low, she can artificially modify her sifted key
  $\SEB$ to increase the error that \Bob\ registers.)
\item \Bob\ performs step (Pb) of the protocol ($\TA \rightarrow
  \TE$), obtains $\hatSB=\SEB$, calculates the error rate, determines
  $M^\Bob=(\nB,\BEB,\errorr/\fail)$, where $\BAB \rightarrow \BEB$ and
  sends message (P2).
\item \Alice\ accepts the authenticity of all the messages she has
  received, i.e.\ (S2), (S4), (P2), since \Eve\ has not modified any
  message and performs step (Pc) sending (P3).
\item[(Pc')] \Eve\ intercepts (P3). To break the authentication of
  (P3), \Eve\ calculates another PA function $\PE$, such that
  $\PE(\SEB)=\KA$ and $t_2=\MAC{K_2}{\BE,\TE,\PE}$. 
  \footnote{In \ref{app:collprob2} we demonstrate that for typical scenarios the probability that in step (Pc') a useful PA function $\PE$ for Eve exists is almost one.}
  To ensure the last
  condition it is sufficient that the message $(\BE,\TE,\PE)= M^\Eve$
  collides with $M^\Alice$ under the inner authentication hash
  function $f$, i.e $f(M^\Eve)=f(M^\Alice)$. \Eve\ sends (P3') to
  \Bob. (If Eve would be satisfied with Alice and Bob having different
  keys, both of which she knows, Eve only searches for any PA function
  $\PE$ such that $f(M^\Eve)=f(M^\Alice)$, but accepts
  $\KB=\PE(\SEB)=\KEB\ne\KA=\PA(\SEA)$.)
\item \Bob\ accepts the authenticity of all the messages he has
  received, i.e.\ (S3'), (P1'), (P3'), since he has received a valid
  tag ($t_2$) and performs the final step of the protocol to get
  $\KB=\PE(\hatSB)=\KE=\KA$.
\end{enumerate}

\end{enumerate}
\newcommand{\setP}{\mathcal{P}}
\newcommand{\bin}{\{0,1\}}
\newcommand{\len}{\mathrm{len}}

\subsection{Protocol 3 -- BB84 with immediate message authentication -- Bob sends bases}\label{sec:prot3}
This protocol is a variant of protocol 1, also using  immediate message authentication. Only part (S), i.e.\ the quantum state transmission and sifting is different: After measuring the quantum signals, instead of sending an acknowledge message as in protocol 1, Bob sends his bases information to Alice (implicitly acknowledging that he has finished his measurements). Alice replies with her basis information.

\subsubsection{State transmission and sifting}\label{sec:prot3-qsift} \ \\
SUMMARY: 2 classical messages with MACs are exchanged.
% RESULT: Alice and Bob share correlated bit strings (sifted keys) $\SA$, $\SB$.
\begin{enumerate}[1.]
 \item \Setup \Alice\ and \Bob\ share two keys $K_1,K_2$.
 \item \Messages
   \begin{enumerate}[(S1)]
   \item $\Alice \stackrel{\qchan}{\longrightarrow} \Bob: \quad \sA$ % (S1)
   \item $\Alice \stackrel{\cchan}{\longleftarrow} \Bob: \quad \msg[K_1]{\BB}$ % (S2)
   \item $\Alice \stackrel{\cchan}{\longrightarrow} \Bob: \quad \msg[K_2]{\BAB}$ % (S3)
   \end{enumerate}

 \item \Actions
 \begin{enumerate}[(Sa)]
\item same as (Sa) in protocol 1: \Alice\ creates two random bit strings, $\IA,\BA \in_r \{0,1\}^N$. For each pair $\left(\IA_k,\BA_k\right)$ \Alice\ generates the corresponding quantum state $\sA_k\in \{\state^0,\state^1,\state^2,\state^3\}$. Using $\qchan$, \Alice\ sends the quantum state $\sA=\bigotimes_{k=1}^N\sA_k$ (``string'' of all $\sA_k$'s), i.e.\ (S1), to \Bob.
\item same as (Sb) in protocol 1: \Bob\ creates a random bit string $\BB\in_r\{0,1\}^N$. \Bob\ measures $\sA$ in bases $\BB$ and obtains $\IB \in \{0,1,\myempty\}^N$ as result. For all $k$ with $\IB_k=\myempty$, \Bob\ sets $\BB_k=\myempty$.
%\Bob\ calculates the string ${\BsB}\in\{0,1,\myempty\}^N$, such that $\BsB_k=\BB_k$ if $\IB_k\in \{0,1\}$, and $\BsB_k=\myempty$ if $\IB_k = \myempty$ ($\BsB$ carries the information whether a detector has clicked, and when in which basis).

\item Using $\cchan$, \Bob\ sends (S2), i.e.\ $\BB$, to \Alice.

\item \Alice\ waits until she has received (S2). \Alice\ calculates the bit string $\BAB$, such that $\BAB_k=1$ if $\BA_k=\BB_k$
%\Alice\ and \Bob\ prepared and measured $\sA_k$ in the same basis
, and $\BAB_k=0$, otherwise.
% if the bases were different, or \Bob\ did not measure a signal.
\Alice\ removes from $\IA$ all bits $\IA_k$ where $\BAB_k=0$ and obtains $\SA$.

\item  Using $\cchan$, \Alice\ sends (S3), i.e.\ $\BAB$, to \Bob.

\item \Bob\ removes from $\IB$ all bits $\IB_k$ where $\BAB_k=0$ and obtains $\SB$.
\end{enumerate}
\end{enumerate}

\subsubsection{Post processing (P)}
%Error correction, confirmation, privacy amplification}
\label{sec:prot3-postproc} \ \\
This part is completely identical to part (P) of protocol 1, cf.\ Sec.~\ref{sec:prot1-postproc}.

\subsubsection{Attack against Protocol 3}\label{sec:prot3-attack} \ \\

Eve replaces the quantum channel between Alice and Bob, with ideal quantum channels and her instrumentation. Eve must be able to prepare and perfectly measure quantum states. She does not need a quantum memory to perform her attack. Essentially this attack is a modified version of the well known intercept-resend attack, whereby the currently discussed authentication mechanism allows Eve to conceal the difference  between the sifted keys of Alice and Bob (of roughly 25\%) in the postprocessing stage of the protocol.

%SUMMARY: Eve uses a quantum memory, prepares quantum states, and modifies sifting messages.\\
\noindent RESULT: Alice, Bob, and Eve share identical keys $\KA=\KB=\KE$.
\begin{enumerate}[1.]
 \item \Notation \\
 $\tildeBases^x$: a string that deviates slightly from $\Bases^x$ to reach a hash collision with a given tag $t$ [used in messages (S2') and (S3')].
% \item System setup: \Alice\ and \Bob\ must provide time synchronization between the state preparation and the measurement

 \item \Attackmessages Let $t_1=\MAC{K_1}{\BB}$, $t_2=\MAC{K_2}{\BAE}$, $t_3=\MAC{K_3}{\TA}$, $t_5=\MAC{K_5}{\PA}$.
\begin{enumerate}[(S1)]
\item $\Alice \stackrel{\qchan}{\longrightarrow} \Eve: \quad \sA$ % (S1)
\item[(S1')] $\Eve \stackrel{\qchan}{\longrightarrow} \Bob: \quad \sE$ % (S2')
\item $\Eve \stackrel{\cchan}{\longleftarrow} \Bob: \quad \BB, t_1$ % (S2)
\item[(S2')] $\Alice \stackrel{\cchan}{\longleftarrow} \Eve: \quad \tildeBE,t_1$ % (S2')
\item $\Alice \stackrel{\cchan}{\longrightarrow} \Eve: \quad \BAE, t_2$ % (S3)
\item[(S3')] $\Eve \stackrel{\cchan}{\longrightarrow} \Bob: \quad \tildeBEB, t_2$ % (S3')
\item[(P1)] $\Alice \stackrel{\cchan}{\longrightarrow} \Eve: \quad \TA, t_3$ % (P1)
\item[(P1')] $\Eve \stackrel{\cchan}{\longrightarrow} \Bob: \quad \TE, t_3$ % (P1')
\item[(P2)] $\Alice \stackrel{\cchan}{\longleftarrow} \Bob: \quad \msg[K_4]{\errorr} \quad / \quad \msg[K_4]{\fail}$ % (P2)
\item[(P3)] $\Alice \stackrel{\cchan}{\longrightarrow} \Eve: \quad \PA, t_5 \quad / \quad \text{-----}$ % (P3)
\item[(P3')] $\Eve \stackrel{\cchan}{\longrightarrow} \Bob: \quad \PE, t_5 \quad / \quad \text{-----}$ % (P3')
  \end{enumerate}

 \item \Attackactions %(\Alice\ and \Bob\ adhere to the protocol of Sec.~\ref{sec:prot3-qsift}. We only state the actions performed by \Eve.
%\\
\begin{enumerate}[(Sa)]
\item \Alice\ performs step (Sa) of the protocol.
\item[(Sa')] \Eve\ creates a random bit strings, $\BE \in_r
  \{0,1\}^N$. \Eve\ intercepts (S1) from \Alice\ and measures $\sA$ in
  bases $\BE$, she obtains $\IE$. For each pair
  $\left(\IE_k,\BE_k\right)$, \Eve\ prepares the corresponding quantum
  state $\sE_k$ and sends (S1') to \Bob.
\item \Bob\ performs step (Sb) of the protocol ($\sA \rightarrow
  \sE$).
\item \Bob\ performs step (Sc) of the protocol, i.e.\ he sends (S2).
\item[(Sd')] \Eve\ intercepts (S2) and performs \Alice's step (Sd) of
  the protocol($\BAB \rightarrow \BEB$, $\BA \rightarrow \BE$) and
  obtains her sifted key with Bob, $\SEB$.
\item[(Sc')] \Eve\ calculates $\tildeBE$, such that the intercepted
  $t_1$ validates $\tildeBE$ and $d_H(\tildeBE,\BE)$ is small. She
  then performs \Bob's step (Sc) of the protocol ($\BB \rightarrow
  \tildeBE$), i.e.\ she sends (S2') to \Alice.
\item \Alice\ performs step (Sd) of the protocol ($\BB \rightarrow
  \tildeBE$, $\BAB \rightarrow \BAE$), she obtains $\BAE$ (which is
  defined by $\BAE_k=1$, if $\BA_k=\tildeBE_k$, and $\BAE_k=0$,
  otherwise) and $\SA$.
\item \Alice\ performs step (Se) of the protocol ($\BAB \rightarrow
  \BAE$), i.e.\ she sends (S3).
\item[(Sf')] \Eve\ intercepts (S3) and performs \Bob's step (Sf) of
  the protocol ($\IB \rightarrow \IE$, $\BAB \rightarrow \BAE$) and
  obtains (approximately) her sifted key with \Alice, $\SEA$. (There
  are small deviations between $\SA$ and $\SEA$ since \Eve\ had to
  send $\tildeBE$ instead of $\BE$).
\item[(Se')] \Eve\ determines the string $\BEB$, such that $\BEB_k=1$,
  if $\BE_k=\BB_k$, and $\BEB_k=0$, otherwise. \Eve\ then calculates
  the string $\tildeBEB$, such that the intercepted $t_2$ validates
  $\tildeBEB$ and $d_H(\tildeBEB,\BEB)$ is small. Now \Eve\ performs
  \Alice's step (Se) of the protocol ($\BAB \rightarrow \tildeBEB$),
  i.e.\ she sends (S3').
\item \Bob\ performs step (Sf) of the protocol ($\BAB \rightarrow
  \tildeBEB$), and obtains his sifted key, $\SB$ (there are small
  deviations between $\SB$ and $\SEB$ since \Eve\ had to send
  $\tildeBEB$ instead of $\BEB$).
\end{enumerate}

Note: Now Eve possesses almost identical copies of Alice's and Bob's keys, respectively: $\SA\approx\SEA$ and $\SEB\approx\SB$ (while $\SA$ and $\SB$ will differ in approximately 25\% of the bits due to Eve's quantum intercept-resend attack). The subsequent steps allow Eve to transform her key  $\SEA$ into $\SA$ and make Bob transform his key $\SB$  into a new key $\hatSB$, which she knows:

\begin{enumerate}[(Pa)]
\item  \Alice\ performs step (Pa) of the protocol, i.e.\ she sends (P1). 
\item[(Pb')] \Eve\ performs \Bob's step (Pb) of the protocol, i.e.\ she intercepts (P1) to learn the syndrome $\ECCode_i(\SA)$, and corrects her sifted key $\SEA$ to $\SA$.
\item[(Pa')] \Eve\ performs \Alice's step (Pa) of the protocol, but modifies her key $\SEB$ such  that $\ECCode^\Eve(\SEB)$ will allow \Bob\ to correct his sifted key to the modified $\SEB$ and that the resulting (P1'), i.e.\ $\TE=(i,\ECCode_i(\SEB),\COCode,\COCode(\SEB))$, is compatible with tag $t_3$. \Eve\ sends (P1'). 
\item \Bob\ performs step (Pb) of the protocol, i.e.\ he corrects his sifted key $\SB$ and obtains $\hatSB$. Now Eve shares $\SA$ with Alice, and ${\hatSB}$ with Bob.
\item \Alice\ performs step (Pc) of the protocol, i.e.\ she determines a privacy amplification function $\PA$, applies it to her sifted key, and obtains $\KA=\PA(\SA)$. \Alice\ sends (P3). 
\item[(Pc')] \Eve\ intercepts (P3) to learn the privacy amplification function $\PA$ and thus \Alice's final key $\KA$. \Eve\ calculates another PA function $\PE$ such that $\PE(\hatSB)=\KA$ and that (P3') is compatible with tag $t_5$. 
\footnote{In \ref{app:collprob2} we demonstrate that for typical scenarios the probability that in step (Pc') a useful PA function $\PE$ for Eve exists is almost one.}
\item \Bob\ performs step (Pd) of the protocol, i.e.\ he applies
  $\PE$ and gets $\KB=\PE(\hatSB)=\KA$.
\end{enumerate}
\end{enumerate}
Again, Eve managed to break the protocol completely, as she knows Alice's and Bob's shared ``secret'' key.

\subsection{Implications of protocol modifications on the presented attacks}
\label{sec:prot-modifications}

\subsubsection{No separate step for transmitting the privacy amplification function}
In \cite[p.~83]{GH00} it has been proposed that the privacy amplification function $\PA$ is not transmitted in a separate protocol step (our step (P3)), but can be constructed from previously exchanged basis information (\cite{Peev2005} uses this method to counter the attack described in \cite{BMQ05}). However, no strict security proof of the resulting protocol has ever been put forward.

For the discussed two-step authentication our attack against protocol 1 still works without step (P3) since we don't attack the post processing step at all. Also the attack against protocol 3 still works without step (P3), but is not so powerful. Since Eve has complete knowledge of the basis information, she can just apply the respective PA function individually to her keys with Alice and Bob. Consequently, Eve will know Alice's and Bob's final keys which will be, however, different.

The case of protocol 2 is slightly more complicated but the outcome is identical to that of protocol 3. In this case the last communication message from Alice to Bob is (P1), and, naturally, it has  to be extended to carry also the authentication tag $t_2=\MAC{K_2}{M^\Alice }$, whereby now $M^\Alice=(\BA,\TA)$. Eve will have to modify her attack. Now she has to look for an error correction syndrome $\TE$, so that $M^\Eve=(\BE,\TE)$, collides with $M^\Alice$ under the inner authentication hash function $f$, i.e $f(M^\Eve)=f(M^\Alice)$. To do so Eve is free to modify her sifted key $\SEB \rightarrow \tildeSEB$, so that $\TE=(i,\ECCode_i(\tildeSEB),\COCode,\COCode(\tildeSEB))$ would ensure the required collision. As in the case of protocol 3 Eve has complete knowledge of the bases of Alice and Bob. She can again apply the respective PA functions independently and obtain the final keys of Alice and Bob, which differ one from the other.

\subsubsection{One-time pad encryption of the error correction syndrome}

Ref.~\cite{Lutkenhaus1999} presented a protocol in which parity bits are encrypted with a one-time pad (using key that is preshared or generated in previous rounds). 
Since Alice and Bob use in addition a (large) key which is not known to Eve, one could expect that attacks will be impossible.
%Although meanwhile, it became clear that from an information theoretic point of view, this is completely equivalent to ``subtracting'' the length of the syndrome during privacy amplification. XXX insert reference (worst case: Renato Renner: private communication)
Nevertheless, we will briefly outline modified attacks against such a protocol.
% Eve performs an attack against sifting as in the previous attacks. 

If Eve uses a quantum memory in her attack she will learn Alice's complete sifted key. Therefore, she can calculate the exact syndrome, that Alice will OTP-encrypt and send. From the plain and encrypted syndrome, Eve gets the one-time pad, encrypts \emph{her} syndrome with it and continues the attack.

If Eve performs an attack without quantum memory, her
% In case of immediate authentication (protocols 1 and 3) Eve's 
and Alice's sifted key will differ in a small number of bits (the Hamming distance $w$ of the two keys), the positions of which are known to Eve. Thus Eve can create the set of all possible sifted keys of Alice of size $2^w$, which is only a very small subset of all possible keys of length approximately $n/2$, and is also smaller than the set of all possible message tags. Then Eve decides randomly to take one element of this set to be Alice's sifted key. Compared to a guess without previous knowledge she could dramatically increase her chances of guessing correctly, although the probability is still quite low, i.e.\ $p=2^{-w}$. Assuming she has guessed correctly, she can now calculate the syndrome that Alice has sent, and thus get also the one-time pad.  She uses it then for encrypting the syndrome that she sends to Bob.

% We like to add that noisy preprocessing \cite{Renes2007}, in combination with a one-time pad encryption of the syndrom, seems to be immune against our attacks even with quantum memory. The introduction of a small random error in Alice's sifted key after the quantum exchange, as Eve's information on

% In this case, since Alice and Bob use in addition a key which is not known to Eve, 

\subsection{Overview of attack approaches for adversaries with and without quantum memory}

%%% start of macros for the entries of the following table %%%
%%%%%%%%%%%%%%%%%%%%%%%%%%%%%%%%%%%%%%%%%%%%%%%%%%%%%%%%%%%%%%
\newcommand{\mycase}[1]{(Case~#1)}

\newcommand{\IAQ}{
\hyperref[sec:prot1-attack]{\emph{Protocol 1--Interleaving attack:}} \newline
\Eve\ stores $\sA$ in quantum memory, \newline
\Eve\ sends random $\sE(\BE,\IE)$ to \Bob, \newline
\Eve\ substitutes $\tildeBE$ for $\BA$, \newline
\Eve\ measures $\sA$ in $\BA$ and learns $\IA$, \newline
\Eve\ calculates $\tildeBAE$ to force \mbox{$\SEA\approx\SEB$}, \newline
\Eve\ substitutes $\tildeBAE$ for $\BEB$.
} % IAQ
\newcommand{\IAQs}{
\fbox{$\SA=\SEA\boldsymbol{\approx} \SEB\approx\SB$\vphantom{$\not\approx$}} \mycase{1}
}

\newcommand{\DAQ}{
\hyperref[sec:prot2-attack]{\emph{Protocol 2--Interleaving attack:}} \newline
\Eve\ stores $\sA$ in quantum memory, \newline
\Eve\ sends random $\sE(\BE,\IE)$ to \Bob, \newline
\Eve\ substitutes $\BE$ for $\BA$, \newline
\Eve\ measures $\sA$ in $\BA$ and learns $\IA$, \newline
\Eve\ listens to $\BEB$ (no substitution!).
}
\newcommand{\DAQs}{
\fbox{$\SA=\SEA\not\approx \SEB=\SB$} \mycase{2}
}

\newcommand{\IAnQ}{
\emph{Intercept-resend attack:} \newline
\Eve\ measures $\sA$ in $\BE$ and gets $\IE$, \newline
\Eve\ sends $\sE(\BE,\IE)$ to \Bob, \newline
\Eve\ substitutes $\tildeBE$ for $\BA$, \newline
\Eve\ substitutes $\tildeBAE$ for $\BAB$.
}
\newcommand{\IAnQs}{
\fbox{$\SA\approx\SEA\not\approx \SEB\approx\SB$} \mycase{3}
}

\newcommand{\DAnQ}{
\emph{One-sided intercept-resend attack:} \newline
\Eve\ measures $\sA$ in $\BE$ and gets $\IE$, \newline
\Eve\ sends $\sE(\BE,\IE)$ to \Bob, \newline
\Eve\ substitutes $\BE$ for $\BA$, \newline
\Eve\ listens to $\BEB$ (no substitution!).
}
\newcommand{\DAnQs}{
\fbox{$\SA\not\approx\SEA\not\approx \SEB=\SB$} \mycase{4}
}

\newcommand{\IBQ}{
\emph{Interleaving attack:} \newline
\Eve\ stores $\sA$ in quantum memory, \newline
\Eve\ sends random $\sE(\BE,\IE)$ to \Bob, \newline
\Eve\ listens to $\BB$ (no substitution!), \newline
\Eve\ measures $\sA$ in $\BB$, \newline
\Eve\ listens to $\BAB$, determines $\SEA$, \newline
\Eve\ substitutes $\tildeBEB$ for $\BAB$.
}
\newcommand{\IBQs}{
\fbox{$\SA=\SEA\not\approx \SEB\approx\SB$} \mycase{3}
}

\newcommand{\DBQ}{
\emph{Interleaving attack:} \newline
\Eve\ stores $\sA$ in quantum memory, \newline
\Eve\ sends random $\sE(\BE,\IE)$ to \Bob, \newline
\Eve\ listens to $\BB$ (no substitution!), \newline
\Eve\ measures $\sA$ in $\BB$, \newline
\Eve\ listens to $\BAB$, determines $\SEA$, \newline
\Eve\ substitutes $\BEB$ for $\BAB$.
}
\newcommand{\DBQs}{
\fbox{$\SA=\SEA\not\approx \SEB=\SB$} \mycase{2}
}

\newcommand{\IBnQ}{
\hyperref[sec:prot3-attack]{\emph{Protocol 3--Intercept-resend attack:}} \newline
\Eve\ measures $\sA$ in $\BE$ and gets $\IE$, \newline
\Eve\ sends $\sE(\BE,\IE)$ to \Bob, \newline
\Eve\ substitutes $\tildeBE$ for $\BB$, \newline
\Eve\ substitutes $\tildeBEB$ for $\BAB$.
}

\newcommand{\IBnQs}{
\fbox{$\SA\approx\SEA\not\approx \SEB\approx\SB$} \mycase{3}
}

\newcommand{\DBnQ}{
\emph{One-sided intercept-resend attack:} \newline
\Eve\ measures $\sA$ in $\BE$ and gets $\IE$, \newline
\Eve\ sends $\sE(\BE,\IE)$ to \Bob, \newline
\Eve\ listens to $\BB$ (no substitution!), \newline
\Eve\ substitutes $\BEB$ for $\BAB$.
}
\newcommand{\DBnQs}{
\fbox{\mbox{$\SA\not\approx\SEA\not\approx \SEB=\SB$}} \mycase{4}
}

\newcommand{\VertAlice}{\rotatebox{90}{\Alice\ sends bases \hspace{4cm}}}
\newcommand{\VertBob}{\rotatebox{90}{\Bob\ sends bases \hspace{2.5cm}}}
\newcommand{\VertQM}{Y}
\newcommand{\VertnoQM}{N}
\begin{table}[p]

\caption{Overview of attacks against the sifting stage of different protocol variants. QM denotes whether Eve uses a quantum memory in her attack. The notation ``Protocol 1--3'' refers to the protocols and corresponding attacks described in full detail above. $\sE(\BE,\IE)$ denotes the quantum state which encodes the bit string $\IE$ in bases $\BE$. $\approx$ denotes that two sifted keys deviate only weakly (error correction can reconcile them). $\not\approx$ denotes a deviation of two sifted keys by typically 25\%. If not otherwise stated, \Eve\ performs sifting with the appropriate bases of \Alice\ and \Bob.}
\label{table:sifting-attacks}
\begin{tabularx}{\textwidth}{l|l|X|X}
\br
 & QM & Immediate Authentication & Delayed Authentication$^\dag$ \\
\mr
\multirow{2}{*}{\VertAlice} & \VertQM & \IAQ & \DAQ \\
& & \IAQs & \DAQs \\ \cmidrule{2-4}
 & \VertnoQM & \IAnQ & \DAnQ \\
& & \IAnQs & \DAnQs \\ \mr
\multirow{2}{*}{\VertBob} & \VertQM & \IBQ & \DBQ \\
& & \IBQs & \DBQs \\ \cmidrule{2-4}
 & \VertnoQM & \IBnQ & \DBnQ \\
& & \IBnQs & \DBnQs \\
\br
\end{tabularx}
\noindent{\footnotesize $^\dag$In these cases \Eve\ does not substitute
messages from \Bob\ to \Alice.}
\end{table}

\newcommand{\PCaseI}{
\IAQs \newline
\noindent
\begin{tabular}{l@{\hspace{.8em}}l}
EC: & \Eve\ does nothing.
\end{tabular} \newline \newline \newline
result: \fbox{$\SA=\SEA=\SB\vphantom{\not\approx}$} \newline
\begin{tabular}{l@{\hspace{.8em}}l}
PA: & \Eve\ listens to the PA function $\PA$, \\
& \Eve\ calculates $\KE:=\PA(\SEA)$. \\
\end{tabular} \newline \newline \newline \newline
result: \fbox{$\KA=\KE=\KB$}
}

\newcommand{\PCaseII}
{
\DAQs \newline
\begin{tabular}{l@{\hspace{.8em}}l}
EC: & \Eve\ intercepts $\TA$, \\
%& (\Eve\ corrects $\SEA$, obtains $\SA$,) \\
& \Eve\ calculates $\TE$, \\
& \Eve\ sends $\TE$ to \Bob. \\
\end{tabular} \newline
result: \fbox{$\SA=\SEA\not\approx \SEB=\SB$} \newline
\begin{tabular}{l@{\hspace{.8em}}l}
PA: & \Eve\ intercepts the PA function $\PA$, \\
& \Eve\ calculates $\KE:=\PA(\SEA)$, \\
& \Eve\ calculates new PA function $\PE$, \\
& \Eve\ sends $\PE$ to \Bob.
\end{tabular} \newline
result:
\fbox{$\KA=\KE=\KB$%
%\noindent
%\begin{tabular}{@{}l@{}}
%$\KA=\KE=\KB$ or\\
%$\KA=\KEA\ne\KEB=\KB$
%\end{tabular}
}
}

\newcommand{\PCaseIII}{
\IAnQs \newline
\begin{tabular}{l@{\hspace{.8em}}l}
EC: & \Eve\ intercepts $\TA$, \\
& \Eve\ corrects $\SEA$, obtains $\SA$, \\
& \Eve\ modifies $\SEB$, calculates $\TE$, \\
& \Eve\ sends $\TE$ to \Bob. \\
\end{tabular} \newline
result: \fbox{$\SA=\SEA\not\approx \SEB=\SB$} \newline
\begin{tabular}{l@{\hspace{.8em}}l}
PA: & \Eve\ intercepts the PA function $\PA$, \\
& \Eve\ calculates $\KE:=\PA(\SEA)$, \\
& \Eve\ calculates new PA function $\PE$, \\
& \Eve\ sends $\PE$ to \Bob.
\end{tabular} \newline
result:
\fbox{%
$\KA=\KE=\KB\vphantom{\ne}$%
%\noindent
%\begin{tabular}{@{}l@{}}
%$\KA=\KE=\KB$ or\\
%$\KA=\KEA\ne\KEB=\KB$
%\end{tabular}
}
}

\newcommand{\PCaseIV}{
\DAnQs \newline \noindent
\begin{tabular}{l@{\hspace{.8em}}l}
EC: & \Eve\ intercepts $\TA$, \\
%& (\Eve\ corrects $\SEA$, obtains $\SA$,) \\
& \Eve\ calculates $\TE$, \\
& \Eve\ sends $\TE$ to \Bob. \\
\end{tabular} \newline \newline \newline
result: \fbox{$\SA\not\approx\SEA\not\approx \SEB=\SB$} \newline
\begin{tabular}{l@{\hspace{.8em}}l}
PA: & \Eve\ intercepts the PA function $\PA$, \\
& \Eve\ calculates $\KE:=\PA(\SEA)$, \\
& \Eve\ calculates new PA function $\PE$, \\
& \Eve\ sends $\PE$ to \Bob.
\end{tabular} \newline
result:
\fbox{%
%\noindent
%\begin{tabular}{@{}l@{}}
$\KA\ne\KE=\KB$%
%\end{tabular}
}
}

\begin{table}[tp]
\caption{Overview of four attack classes against the protocol stages after sifting. The attacks pertain to the output of sifting, which according to Table \ref{table:sifting-attacks}, yields four different types of correlations between the sifted keys of \Alice, \Eve, and \Bob: two for immediate (cases 1, 3) and two for delayed authentication (cases 2, 4), respectively. Note, that for the sake of simplicity we do not use the ``hat'' notation for error corrected keys.}
\label{table:post-processing-attacks}
\begin{tabularx}{\textwidth}{X|X}
\br
Immediate Authentication & Delayed Authentication$^\dag$ \\
\mr
\PCaseI & \PCaseII \\
\mr
\PCaseIII & \PCaseIV \\
\br
\end{tabularx}
\noindent{\footnotesize $^\dag$In these cases \Eve\ does not substitute 
messages from \Bob\ to \Alice.}
\end{table}
%

%\begin{center}
%\begin{table}
%\begin{tabular}{lllll}
%Attack & Authentication & Bases Choice & Quantum Memory & Achievement\\
%\hline
%1 & Immediate & \Alice & needed & $\KA = \KE=\KB$ \\
%2 & Delayed & \Alice  & needed & $\KA = \KE=\KB$ (*) \\
%3 & Immediate & \Bob & not needed & $\KA = \KE=\KB$ (*)
%\end{tabular}
%\caption{Summary of attacks studied against different variants of the BB84 protocol. \Alice(lice) and \Bob(ob) denote which party announces its basis choice first. (*) We discuss also modifications of attacks 2 and 3 that give Alice and Bob different keys, both known to Eve: $\KEA=\KA \ne \KB=\KEB$.}
%\end{table}
%\end{center}
%

%\subsubsection{Constructing attacks against sets of protocols}

Up to now we have presented three attacks in which Eve on receiving a protocol message from Alice (Bob) sends either the original message or a modified one to Bob (Alice). In Sec.~\ref{sec:2ndattack2} we will present a different kind of attack. The attacks presented so far are not isolated cases of adversary success strategies in the case of weak authentication that uses the approach of Ref.~\cite{Peev2005}. The attacks are actually made up of building blocks that can be combined and applied in a wide variety of settings. We illustrate this fact by presenting a systematic overview of successful attacks against a range of protocols comprising the cases of sifting being started by Alice or Bob, authentication being immediate or delayed. Moreover for all the cases we distinguish between two levels of adversary resources: i) ``classical only'', i.e.\ sufficiently high computing power or ii) ``quantum and classical'', i.e.\ a combination of quantum resources (quantum memory) and classical ones (as in i)).
These attacks are summarized in Tables \ref{table:sifting-attacks} and \ref{table:post-processing-attacks}. %The eavesdropper efforts yield four different outcomes during sifting (Table  \ref{table:sifting-attacks}) for which reason the subsequent post-processing steps are considered for these outcomes alone (Table \ref{table:post-processing-attacks}}. 
The attacks are not described in full detail and the tables focus on the adversary activities alone. The full attacks,  can however be easily deduced by comparing the table contents referring to Attacks 1, 2 and 3 with the detailed description for these cases, given above.

Furthermore, using arguments similar to those presented in Section \ref{sec:prot-modifications} one can construct attacks against modified versions of these protocols, including encryption of error-correction information and reuse of common, sifting-stage randomness for privacy amplification without communication.

%%%%%%%%%%%% new attack 2 %%%%%%%%%%%
\subsection{Another attack against Protocol 2 (Eve attacks in both directions)}\label{sec:2ndattack2}

In our previous attacks Eve substitutes certain messages but sticks to the original message order of the protocol.
In the following attack
Eve exchanges a sequence of messages with Alice first. When she needs to send an authentication tag to Alice, she starts her communication with Bob and continues until she obtains the necessary tag from him. Then Eve continues her communication with Alice.

In contrast to the previous attack against protocol 2 (cf.~Sec.~\ref{sec:prot2-attack}) this attack allows Eve to modify also messages that are sent to Alice.
%(S1) a->e
%(S2') a<-e
%(S3) a->e
%(S1) e->b
%(S2) e<-b
%(S3) e->b
%(S4) e<-b
%(P1') e->b
%(P2) e<-b
%(S4') a<-e
%(P1) a->e
%(P2) a<-e
%(P3) a->e
%(P3') e->b

\begin{enumerate}[1.]
\item \Attackmessages Let $t_1:=\MAC{K_1}{M^\Bob}$, and remember that $t_2:=\MAC{K_2}{M^\Alice}$.
  \begin{enumerate}
  \item[(S1)] $\Alice \stackrel{\qchan}{\longrightarrow} \Eve: \quad \sA$
  \item[(S2')] $\Alice \stackrel{\cchan}{\longleftarrow} \Eve: \quad n^\Eve$
  \item[(S3)] $\Alice \stackrel{\cchan}{\longrightarrow} \Eve: \quad \BA$
  \item[(S1')] $\Eve \stackrel{\qchan}{\longrightarrow} \Bob: \quad \sE=\sA$
  \item[(S2)] $\Eve \stackrel{\cchan}{\longleftarrow} \Bob: \quad n^\Bob$
  \item[(S3')] $\Eve \stackrel{\cchan}{\longrightarrow} \Bob: \quad \BE=\BA$
  \item[(S4)] $\Eve \stackrel{\cchan}{\longleftarrow} \Bob: \quad \BEB$
  \item[(P1')] $\Eve \stackrel{\cchan}{\longrightarrow} \Bob: \quad \TE$
  \item[(P2)] $\Eve \stackrel{\cchan}{\longleftarrow} \Bob: \quad {\errorr}\ /\ \fail,t_1$
  \item[(S4')] $\Alice \stackrel{\cchan}{\longleftarrow} \Eve: \quad \tilde{\Bases}^{\Eve=\Bob}$
  \item[(P1)] $\Alice \stackrel{\cchan}{\longrightarrow} \Eve: \quad \TA$
  \item[(P2')] $\Alice \stackrel{\cchan}{\longleftarrow} \Eve: \quad {\errorr}\ /\ \fail,t_1$
  \item[(P3)] $\Alice \stackrel{\cchan}{\longrightarrow} \Eve: \quad \msg{\PA},t_2\ / \ \text{-----}$
  \item[(P3')] $\Eve \stackrel{\cchan}{\longrightarrow} \Bob: \quad \msg{\PE},t_2\ / \ \text{-----}$
  \end{enumerate}
 \item \Attackactions 
   \begin{enumerate}
   \item[(Sa)] \Alice\ performs step (Sa) of the protocol (prepares $\sA$ and sends it in (S1)).
   \item[(Sc')] \Eve\ intercepts (S1) from \Alice\ and stores $\sA$ in her quantum memory. \Eve\ sends an arbitrary number $n^\Eve$ (S2') to $\Alice$ to trigger $\Alice$'s next message.
   \item[(Sd)] \Alice\ performs step (Sd) of the protocol: she sends (S3), i.e.~$\BA$.
   \item[(Sd')] \Eve\ intercepts (S3), measures $\sA$ in \Alice's preparation bases $\BA$, and obtains $\Alice$'s rawkey $\IA$.
   \item[(Sa')] Using $\IA$ and $\BA$, \Eve\ prepares an identical copy of $\sA$ and sends it (S1') to \Bob.
   \item[(Sb)], (Sc), (Sd''), (Se), (Sf') \Eve\ (instead of \Alice) and \Bob\ follow the protocol--whereby sending (S2), (S3'), (S4)--until they obtain their sifted keys $\SE\approx\SB$.
   \item[(Pa')], (Pb) \Eve\ (instead of \Alice) and \Bob\ follow the protocol--whereby sending (P1'),(P2)--and reconcile their sifted keys. 
   \end{enumerate}
On receiving (P2) \Eve\ has learned $M^\Bob$ and the tag $t_1$ and can now continue her communication with \Alice.
   \begin{enumerate}
   \item[(Se')] \Eve\ calculates a message $\tildeBEB$ such that (i) it is close to $\BEB$ and (ii) $M^{\Alice \leftarrow \Eve}:=(\nE,\tildeBEB,\errorr / \fail)$ collides with $M^\Bob$ under the inner hash function $f$, i.e.~$f(M^{\Alice \leftarrow \Eve})=f(M^\Bob)$. \Eve\ sends $\tildeBEB$ to \Alice\ (S4').
   \item[(Sf)], (Pa) \Alice\ calculates her sifted key $\SA$, and sends (P1). 
   \item[(Pb')] \Eve\ intercepts (P1) and can correct small errors introduced during quantum storage or measurement of $\sA$. Using the original tag $t_1$, \Eve\ forwards (P2')=(P2) to \Alice.
   \item[(Pc)] Since $f(M^{\Alice \leftarrow \Eve})=f(M^\Bob)$, \Alice\ accepts the message as authentic, calculates $\PA$ and $\KA=\PA(\SA)$, and sends (P3) with tag $t_2$.
   \item[(Pc')] \Eve\ calculates a PA function $\PE$ (and obtains $\KE=\PE(\hatSB)$) such that (i) $\PE(\hatSB)=\KA$, and (ii) $M^{\Eve \rightarrow \Bob}:=(\BE,\TE,\PE)$ collides with $M^{\Alice}$ under $f$. \Eve\ calculates $\PE(\hatSB)$, and sends (P3') with tag $t_2$ to \Bob.
   \item[(Pd)] \Bob\ calculates $\KB=\PE(\hatSB)$. 
\end{enumerate}
Eve shares a common ``secret'' key with Alice and Bob. 
In case that \Eve\ cannot achieve condition (i) in step (Pc') she will get two individual keys with \Alice\ and \Bob.
In both cases, protocol 2 is completely broken by the presented attack.
\end{enumerate}

\subsection{Discussion of attacks}

The degree of success of the eavesdropper varies from protocol to protocol and ranges from a complete three party identity of the generated key -- $\KA=\KE=\KB$, to ``separate worlds'' outcome -- $\KA=\KEA \neq \KEB=\KB$ (e.g. in a case of privacy amplification with no communication), to a successful attack over one of the legitimate parties (calling for a subsequent isolation of the other)-- i.e.\ $\KA \neq \KE=\KB$. Moreover the success can be achieved either deterministically  or sometimes only probabilistically as in certain cases of encrypted transmission of error correction information.

%\subsection{}

This analysis underlines what was already mentioned in Section \ref{sec:intro-attacks}. As the attack mechanism fundamentally requires finding hash collisions of the internal authentication function that are \emph{useful} to the eavesdropper, the different protocol versions discussed above, allow inequivalent optimal adversarial approaches.
As it is to be expected, the  availability of quantum resources simplifies the task of the eavesdropper but does not automatically lead to more powerful attacks. On the other hand immediate authentication also provides a leverage to the attacker as she does not have to correlate all her actions across the post-processing chain. 
%If the protocol makes hash collisions in the last message from a party unlikely (as in 
This gives the somewhat counter-intuitive observation, that \emph{fewer authentication tags} result in more difficulty for the attacker if he wants to keep the original message order! Furthermore sifting initiated by Bob also poses more difficulties to Eve as she can not learn the full information of Alice as is in the opposite case. Finally if part of the postprocessing information remains unknown to the eavesdropper, as in the case of encrypted reconciliation, then a deterministically successful attack strategy is not always guaranteed.

With all this said it must be underlined that Eve can find \emph{useful collisions} only if she can fake the protocol communication by hiding her modifications in the typically available random degrees of freedom. If such are unavailable or strongly reduced (as e.g. in the case of protocols with delayed authentication or with communication-less privacy amplification) the room for attack is narrowed resulting in a number of cases in ``separate world'' or even ``one-sided'' adversarial success. Still in all discussed cases there always exists an attack strategy that renders the corresponding protocol version insecure.

%Advantages Disadvantages of Immediate/Delayed Authentication
%Very short and freedom less messages make the attacks difficult may be in some versions impossible
%The special case of Attack - ready after sifting
%on sufficient and necessary conditions ASU sufficient , necessary - for conclusions
%
%reason for success - finding collisions , that might be hidden in existing degrees of freedom of the protocol

%%% Local Variables: 
%%% mode: latex
%%% TeX-master: "NJPmain"
%%% End: 

\section{Countermeasures}\label{sec:Countermeasures}

We will now propose a countermeasure that mitigates or, at a cost,
prohibits the attacks exemplified in the previous section. One could
consider encrypting parts of the communication between Alice and Bob
\cite{Abidin2009,Lutkenhaus1999}, but we will concentrate on
strengthening the two-step authentication below. As we shall see,
there are a number of possibilities ranging from increasing Eve's need
for large computational power, all the way to information-theoretic
security.  As can be expected, the cost of this security improvement
comes in the form of an increased secret key consumption.

Let us first consider the main enabler of the attacks presented in the
previous section. The reason that the attacks are possible is that
when Eve receives (or intercepts) Alice's message, she can immediately
check if her message $m^\Eve$ coincides with Alice's under the
publicly known hash function $f$. If not, Eve is free to choose
another message $\tilde{m}^\Eve$ that does coincide with Alice's under
$f$, although in some situations there is a small price to pay as
described above.  To prohibit this we should make it difficult or
impossible for Eve to check for this coincidence.  The essence of our
proposed countermeasure is to use an extra bitsequence to make the
output of the public hash function difficult to predict, or even
secret, to Eve.  This is done in the following way: prepend an extra
bitsequence $S$ to the message and authenticate the result.  Instead
of using the tag $t=g_K(m)=h_K(f(m))$, use the tag
$t=g_K(S||m)=h_K(f(S||m))$. If, for example, $S$ is random and secret
to Eve, then the output $f(S||m)$ will also be secret to Eve, and she
will not be able to search for coincidences in the above manner.

It should be stressed that $S$ should be prepended to the message
before applying $f$. The bitsequence $S$ should \emph{not} be
concatenated with $f(m)$. The reason for this is fairly obvious. If
$S$ is concatenated with $f(m)$ so that $t=h_K(S||f(m))$ or
$t=h_K(f(m)||S)$, then Eve can still apply her original attack
strategy. All Eve needs in this case is still to find a message that
collides with Alice's message under $f$.  We should also stress that
for certain classes of hash functions, prepending $S$ to the message
has advantages over appending to $m$ (so that $t=h_K(f(m||S))$).  When
using iterative hash functions like \textbf{SHA-1} to calculate
$f(m||S)$, Eve can ignore $S$ and search instead for a message
$m^\prime$ such that $f(m^\prime)=f(m)$. This is known as a
partial-message collision attack, see Chapter 5 in Ref.~\cite{FSK}.
If $f$ is computed iteratively, $f(m^\prime)=f(m)$ will automatically
give $f(m^\prime||S)=f(m||S)$ (with appropriate block lengths).
This is prohibited by  prepending $S$ instead, to use $f(S||m)$.

Of course, a random secret $S$ would consume secret key, and this may
not be desirable. Selecting $S$ can be done in a few ways, and these
are the alternatives (including a random secret $S$):
\begin{description}
\item[A salt,] a random but fixed public bitstring, per device or per
  link. This would not make Eve's task much harder, but it would help
  a little in certain situations: for some messages, such as
  preparation and/or measurement settings, Eve does not need to use a
  random bitstring. She can use a fixed (random-looking) bitstring and
  for that message, a pre-calculated table of messages with low
  Hamming distance and their corresponding intermediate tags
  \cite{Abidin2009}.  Even though a full table might have an excessive
  number of entries ($2^{256}$ is a large number), a partial table
  could ease Eve's calculational load (as in a rainbow table), or
  alternatively increase her probability of success.  A salt would
  force Eve to create the table anew for each device or link.
\item[A nonce,] a random public bitstring, per authentication
  attempt. This may seem like a big improvement because it seems Eve
  cannot use a pre-calculated table, forcing her to make the
  calculations online.  However, the nonce needs to be transmitted
  from Alice to Bob or vice versa, and is not separately
  authenticated, since this would need secret key better used
  elsewhere.  A nonce can therefore be changed in transit by Eve, and
  this increases her possibilities. Authenticating a message from
  Alice to Bob, there are two sub-alternatives:
  \begin{enumerate}[a)]
  \item The nonce is generated by Alice and sent to Bob together with
    the tag, and Eve can change it in transit.
  \item The nonce is generated by Bob and sent to Alice after he has
    received the message. One alternative for Eve is to commit to a
    message so that she can receive the nonce from Bob, and then
    change the nonce in transit. In effect, she can now change Alice's
    message since that contains the nonce.
  \end{enumerate}
  In both cases, Eve needs to find a collision online, but Eve now has
  a message part that she can change to any value she desires.
  Therefore, her attack is easier in this setup, not more difficult.
\item[Fixed secret key,] a random but fixed secret bitstring,
  per device or per link.  In this case, Eve cannot apply the previous
  attack on the authentication, because she cannot check for
  collisions directly since $f(S||m)$ is secret to her. To search for a
  message $m^\prime$ useful to Eve (i.e., having low Hamming
  distance to $\mE$) such that $f(S||m^\prime) = f(S||m)$ has maximal
  probability (given the distribution of $S$) is computationally very
  costly.  Moreover, we expect this maximal probability to be very
  low, but an upper bound is difficult to obtain and depends on
  details of the
  hash function, see below.

  As regards using a fixed secret \cite{AJA}, if Eve has partial
  knowledge, no matter how small, on the secret key $K$ identifying
  $h_K$, this information will accumulate over the rounds as
  information on $S$.  Remember that after the initial pre-shared key
  is used up, $K$ will consist of QKD-generated key that is
  $\epsilon$-perfect (the trace distance between the probability
  distribution of the key and the uniform distribution is $\epsilon$),
  where $\epsilon$ is nonzero. Therefore, after a large number of rounds,
  this reduces to using a random fixed public bitstring (salt) as
  discussed above.
\item[Secret key,] a random secret bitstring, per
  authentication attempt.  Here also, Eve cannot apply the previous
  attack on the authentication, because she cannot check for
  collisions directly since $f(S||m)$ is secret to her.  
  The situation is almost identical to the fixed secret key case but Eve's
  task is even harder as she cannot accumulate information on $S$.
\end{description}

This countermeasure is simple to implement, and the last alternative
above seems preferable, if only the key consumption is low. Choosing
$S$ to be of the same size as the tag gives a high computational load
on Eve, and is efficient in terms of key consumption. It is, however,
difficult to estimate the probability of success for Eve, if she has
large computational power.

Let us examine what conditions need to be fulfilled to make the
two-step authentication ITS. If the last alternative above is used, it
is clear that we want a low probability of collision for a random
value of $S$. And this is obtained if two distinct messages collide
under $f$ only for a small number of values of $S$. More formally, let
$\mS$ be the set of values of $S$. Then, if for any two distinct $m_1$,
$m_2\in\mM$ $|\{S\in\mS: f(S||m_1) = f(S||m_2)\}|\leq
\epsilon'|\mS|$, we automatically have a low collision probability.  A
close look at the above condition would tell us that it is precisely
the condition for a family of hash functions indexed by $S$ to be
$\epsilon'$-AU$_2$ (see \ref{uhf}).  The following theorem states that
this condition is necessary and sufficient to restore ITS.

%%%%%%%%%%%%%%%%%%%%%%%%%%%%%%%%%%%%%%%%%%%%%%%%%%%%%%%%%%%%%%%%%%%%
%%%%%%% Theorem 3 in combinatorial form %%%%%%%%%%%%%%%%%%%%%%%%%%%%
%%%%%%%%%%%%%%%%%%%%%%%%%%%%%%%%%%%%%%%%%%%%%%%%%%%%%%%%%%%%%%%%%%%%
\begin{theorem}\label{thm:fix}
  Let $\mM$, $\mZ$ and $\mT$ be finite sets. Let $\mF$ be a family of
  hash functions from $\mM$ to $\mZ$, $\mH$ a family of SU$_2$ hash
  functions from $\mZ$ to $\mT$, and $\mG := \mH\circ\mF$, where
  $\circ$ stands for element-wise composition. Then $\mG$ is
  $\epsilon$-ASU$_2$ if and only if $\mF$ is $\epsilon'$-AU$_2$, where
  $\epsilon=\epsilon'(1-1/|\mT|)+1/|\mT|$.
\end{theorem}

The proof can be found in \ref{uhf}. Thus, to make the two-step
authentication ITS, we should construct our fixed public hash function
$f$ with the help of an AU$_2$ hash function family $\mF$ as follows:
\begin{equation}
  f(S||m)=f_S(m),\ f_S\in\mF.
\end{equation}
In words, $f$ separates $S$ from the concatenation $S||m$ and uses it
as index to select from the hash function family $\mF$ an individual
member $f_S$ which is applied to the original message $m$.

Theorem~\ref{thm:fix} makes it possible to relate the message length
$\log|\mM|$, the security parameter $\epsilon'$, and the key
consumption of the system. Let us aim for a final $\epsilon$-ASU$_2$
family with $\epsilon=2/|\mT|-1/|\mT|^2$, i.e., $\epsilon'=1/|\mT|$.
Then, the bound by Nguyen and Roscoe \cite{NR} is tight:
\begin{equation}
  |\mF|>|\mT|\big\lceil \log|\mM|/\log|\mZ| -1 \big\rceil.
\label{eq:9}
\end{equation}
In \cite{NR} there are two lower bounds, but both can be written in
this way. The bound applies when
$\epsilon'|\mZ|>1+\log|\mZ|/(\log|\mM|-\log|\mZ|)$. 
% and here we have$\log|\mM|>2\log|\mZ|$ and $\epsilon'|\mZ|>3/2$.  
The optimal family
\cite{NR} is that of polynomial evaluation hashing over finite fields
\cite{Boer,Johansson,Taylor}. Therefore, using polynomial hashing with
$|\mF|=|\mZ|$, we can authenticate messages as long as
\begin{equation}
  \log|\mM|<\Big(\frac{|\mZ|}{|\mT|} + 1\Big)\log|\mZ|.
\end{equation}
For example, if $|\mZ| = 2^{256}$, $|\mT|=2^{64}$ and
$\epsilon=2^{-63}-2^{-128}$, then messages of length
$\log|\mM|<2^{200}\approx10^{60}$ bits can be authenticated. The
second step of the authentication uses an SU$_2$ hash function
$\mZ\to\mT$, which needs a key of length at least $\log|\mZ|
+\log|\mT|$ bits \cite{Stinson92,Bierbrauer1}. Thus, the total
required key length is $2\log|\mZ| + \log|\mT|$, in this case $576$
bits. By adjusting $|\mZ|$ to the maximum message length, this scheme
can authenticate one terabit (petabit, exabit, zettabit, yottabit) of
data using 260 (280, 298, 318, 338) bits of secret key.

This construction makes the two-stage authentication ITS at the price
of increasing the key consumption slightly. There are other efficient
constructions of ASU$_2$ hash functions as well, see e.g.,
\cite{Krawczyk94, Bierbrauer1, Bierbrauer2, AJA1}. Some of these
authenticate message of arbitrary length with fixed key consumption at
the price of a varying $\epsilon$, while others have fixed $\epsilon$
but varying key consumption. They also vary in terms of their
computational speed.  The numbers are in the same range as the above
presented ITS authentication, and all mentioned schemes are reasonably
efficient.

%%% Local Variables:
%%% mode: latex
%%% TeX-master: "NJPmain"
%%% End:

\section{Conclusions}\label{sec:Conclusions}

The main conclusion of our extensive analysis is: \emph{do not use non-ITS authentication in QKD
if you want to achieve ITS security}. This may sound rather obvious
but nevertheless in our oppinion it is always good to know what exactly goes wrong
if you break the rules.

So, we have presented a comprehensive case study of attacks
that compromise QKD in the non-ITS authentication setting of 
\cite{Peev2005}, that creates message tags by composing an inner public hash function with an outer function from a strongly universal hashing (SU$_2$) family.
From the point of view of the attacker, who is equipped with unbounded
computing resources, this composition has the following properties:
(i) inserting a randomly chosen message or substituting messages with
a randomly chosen message is as hard as in the SU$_2$ case and thus
cannot be used in attacks, (ii) but more interestingly, substituting a
message with another that collides under the public hash function will
always work. As has been shown previously \cite{Peev2005} property (i)
does prohibit straightforward MITM attacks (cf.~Definition
\ref{def:SF-MITM}).

The sophisticated MITM attacks dicussed here capitalize on property
(ii) to successfully target many QKD protocol versions: protocols that
use individual authentication of each message, or that use delayed
authentication of all messages, protocols where Bob sends an
acknowledgement message to trigger Alice's sifting message (containing
her bases choice), or where Bob directly sends his bases choice, see
Tables \ref{table:sifting-attacks} and
\ref{table:post-processing-attacks}. All the attacks are enabled by
the fact that the number of messages that collide with a given
protocol message (or sequence of messages) of typically at least
several hundred bits is extremely huge. Therefore, almost certainly 
(see Sec.~\ref{sec:HashCollisions}) 
there exists at least one colliding message that allows the
eavesdropper to perform her attack. In some attacks Eve needs less
computing resources if she possesses quantum memory.

We stress that the discussed attack pattern is not
restricted to one single instance, the specific authentication
mechanism of Ref.~\cite{Peev2005} that we study here. 
We conjecture, that whenever property (ii) holds, 
i.e.~collisions can be found, and the protocol does not use additional 
secret key \cite{Abidin2009,Peev2009} (e.g.\ for encryption of messages) the 
adversary can compromise the security of the key generated 
by QKD, following an interleaving approach along the lines 
of that discussed in this paper.

The countermeasures discussed in this paper use more secret key,
specifically to prevent finding collisions.  Prepending secret key
material to the message, before applying the public hash function,
will increase the computational resources needed for a successful
attack substantially, at a low cost in terms of key
material. 

Furthermore, we can achieve 
Universally-Composable Information-Theoretic Security of the 
authentication scheme of \cite{Peev2005} by 
replacing the publicly known hash function with an Almost Universal$_2$ 
function family. 
This requirement is necessary and sufficient for ITS
of the two step authentication; the necessity of this condition is also a new result
of this paper.

%%% Local Variables: 
%%% mode: latex
%%% TeX-master: "NJPmain"
%%% End: 

\begin{acknowledgements}
This work has been supported by the Vienna Science and Technology
Fund (WWTF) via project ICT10-067 (HiPANQ) and also partly by the
Austrian Research Promotion Agency (FFG) within the project Archistar
(Bridge-2364544).
\end{acknowledgements}

\bibliographystyle{spbasic}      % basic style, author-year citations
\bibliography{qkd}

\begin{thebibliography}{33}
\providecommand{\natexlab}[1]{#1}
\providecommand{\url}[1]{{#1}}
\providecommand{\urlprefix}{URL }
\expandafter\ifx\csname urlstyle\endcsname\relax
  \providecommand{\doi}[1]{DOI~\discretionary{}{}{}#1}\else
  \providecommand{\doi}{DOI~\discretionary{}{}{}\begingroup
  \urlstyle{rm}\Url}\fi
\providecommand{\eprint}[2][]{\url{#2}}

\bibitem[{Abidin and Larsson({2009})}]{Abidin2009}
Abidin A, Larsson J{\AA} ({2009}) Vulnerability of ``{A} novel
  protocol-authentication algorithm ruling out a man-in-the-middle attack in
  quantum cryptography{''}. {International Journal of Quantum Information}
  {7}({5}):{1047--1052}

\bibitem[{Abidin and Larsson(2011)}]{AJA}
Abidin A, Larsson J{\AA} (2011) Security of authentication with a fixed key in
  quantum key distribution, \urlprefix\url{http://arXiv.org/abs/1109.5168v1}

\bibitem[{Abidin and Larsson(2012)}]{AJA1}
Abidin A, Larsson J{\AA} (2012) New universal hash functions. In: Lucks S,
  Armknecht F (eds) WEWoRC 2011, Springer-Verlag, LNCS, vol 7242, pp 99--108

\bibitem[{Ben-Or and Mayers(2004)}]{BenOr2004}
Ben-Or M, Mayers D (2004) {General Security Definition and Composability for
  Quantum \& Classical Protocols}
  \urlprefix\url{http://arXiv.org/abs/quant-ph/0409062}

\bibitem[{Ben-Or and Mayers(2005)}]{BenOr2005}
Ben-Or M, Mayers D (2005) {The Universal Composable Security of Quantum Key
  Distribution}. In: Kilian J (ed) Proceedings of TCC 2005, Cambridge, MA, USA,
  Springer 2005, Lecture Notes in Computer Science, vol 3378, pp 386--406,
  \urlprefix\url{http://arXiv.org/abs/quant-ph/0409078}

\bibitem[{Bennett and Brassard(1984)}]{BB84}
Bennett CH, Brassard G (1984) Quantum cryptography: Public key distribution and
  coin tossing. In: Proc. of the {IEEE} Int. Conf. on Computers, Systems, and
  Signal Processing, {IEEE} New York, Bangalore, India, 175{\textendash}179

\bibitem[{Bennett et~al(1992)Bennett, Bessette, Brassard, Salvail, and
  Smolin}]{BennettBBSS92}
Bennett CH, Bessette F, Brassard G, Salvail L, Smolin JA (1992) Experimental
  quantum cryptography. J Cryptology 5(1):3--28

\bibitem[{Beth et~al(2005)Beth, M\"uller-Quade, and Steinwandt}]{BMQ05}
Beth T, M\"uller-Quade J, Steinwandt R (2005) Cryptanalysis of a practical
  quantum key distribution with polarization-entangled photons. Quantum
  Information and Computation 5(3):181--186

\bibitem[{Bierbrauer(1997)}]{Bierbrauer2}
Bierbrauer J (1997) Universal hashing and geometric codes. Designs, Codes and
  Cryptography 11:207--221

\bibitem[{Bierbrauer et~al(1994)Bierbrauer, Johansson, Kabatianskii, and
  Smeets}]{Bierbrauer1}
Bierbrauer J, Johansson T, Kabatianskii G, Smeets B (1994) On families of hash
  functions via geometric codes and concatenation. In: Stinson D (ed) CRYPTO
  '93, Springer-Verlag, Lecture Notes in Computer Science, vol 773, pp 331--342

\bibitem[{den Boer(1993)}]{Boer}
den Boer B (1993) A simple and key-economical unconditional authentication
  scheme. J Comp Sec 2:65--72

\bibitem[{Carter and Wegman(1979)}]{CarterWegman79}
Carter JL, Wegman MN (1979) J Comput Syst Sci 18:143

\bibitem[{Ferguson et~al({2010})Ferguson, Schneier, and Kohno}]{FSK}
Ferguson N, Schneier B, Kohno T ({2010}) {Cryptography Engineering}. {Wily
  Publishing, Inc.}

\bibitem[{Gilbert and Hamrick(2000)}]{GH00}
Gilbert G, Hamrick M (2000) {Practical Quantum Cryptography: A Comprehensive
  Analysis (Part One)}, \urlprefix\url{http://arXiv.org/abs/quant-ph/0009027v5}

\bibitem[{Hayashi(2011)}]{Hayashi2011}
Hayashi M (2011) Exponential decreasing rate of leaked information in universal
  random privacy amplification. IEEE Trans Inf Theory 57:3989--4001

\bibitem[{Hoeffding(1963)}]{Hoeffding63}
Hoeffding W (1963) {Probability Inequalities for Sums of Bounded Random
  Variables}. J Amer Statistical Assoc 58(301):13--30

\bibitem[{Johansson et~al(1993)Johansson, Kabatianskii, and Smeets}]{Johansson}
Johansson T, Kabatianskii G, Smeets B (1993) On the relation between a-codes
  and codes correcting independent errors. In: Advances in Cryptology,
  EUROCRYPT 1993, Springer-Verlag, Lecture Notes in Computer Science, vol 765,
  pp 1--11

\bibitem[{Krawczyk(1994)}]{Krawczyk94}
Krawczyk H (1994) Lfsr-based hashing and authentication. In: Desmedt Y (ed)
  CRYPTO '94, Springer 1994, Lecture Notes in Computer Science, vol 839, pp
  129--139

\bibitem[{L\"utkenhaus(1999)}]{Lutkenhaus1999}
L\"utkenhaus N (1999) Estimates for practical quantum cryptography. Phys Rev A
  59(5):3301--3319, \doi{10.1103/PhysRevA.59.3301}

\bibitem[{Mehlhorn and Vishkin(1984)}]{Mehlhorn84}
Mehlhorn K, Vishkin U (1984) Randomized and deterministic simulations of prams
  by parallel machines with restricted granularity of parallel memories. Acta
  Inf 21:339--374

\bibitem[{Menezes et~al(1996)Menezes, van Oorschot, and Vanstone}]{MenezesOV96}
Menezes A, van Oorschot PC, Vanstone SA (1996) Handbook of Applied
  Cryptography. CRC Press

\bibitem[{M{\"u}ller-Quade and Renner(2009)}]{JMQRR08}
M{\"u}ller-Quade J, Renner R (2009) Composability in quantum cryptography. New
  Journal of Physics 11(8):085,006

\bibitem[{Nguyen and Roscoe(2010)}]{NR}
Nguyen LH, Roscoe AW (2010) New combinatorial bounds for universal hash
  functions. Universal Computing

\bibitem[{Peev et~al({2005})Peev, N\"olle, Maurhardt, Lor\"unser, Suda, Poppe,
  Ursin, Fedrizzi, and Zeilinger}]{Peev2005}
Peev M, N\"olle M, Maurhardt O, Lor\"unser T, Suda M, Poppe A, Ursin R,
  Fedrizzi A, Zeilinger A ({2005}) {A novel protocol-authentication algorithm
  ruling out a man-in-the middle attack in quantum cryptography}.
  {International Journal of Quantum Information} {3}({1}):{225--231}

\bibitem[{Peev et~al({2009})Peev, Pacher, Lor\"unser, N\"olle, Poppe, Maurhart,
  Suda, Fedrizzi, Ursin, and Zeilinger}]{Peev2009}
Peev M, Pacher C, Lor\"unser T, N\"olle M, Poppe A, Maurhart O, Suda M,
  Fedrizzi A, Ursin R, Zeilinger A ({2009}) {Response to ``Vulnerability of `A
  novel protocol-authentication algorithm ruling out a man-in-the-middle attack
  in quantum cryptography'{''}}. {International Journal of Quantum Information}
  {7}({7}):{1401--1407}

\bibitem[{Portmann(2014)}]{Portmann2012}
Portmann C (2014) Key recycling in authentication. {IEEE Trans Inf Theory}
  60(8):4383--4396

\bibitem[{Renner and K{\"o}nig(2005)}]{Renner_Koenig05}
Renner R, K{\"o}nig R (2005) Universally composable privacy amplification
  against quantum adversaries. In: Kilian J (ed) Proceedings of TCC 2005,
  Cambridge, MA, USA, Springer 2005, Lecture Notes in Computer Science, vol
  3378, pp 407--425

\bibitem[{Sasaki et~al(2011)Sasaki, Fujiwara, Ishizuka, Klaus, Wakui, Takeoka,
  Miki, Yamashita, Wang, Tanaka, Yoshino, Nambu, Takahashi, Tajima, Tomita,
  Domeki, Hasegawa, Sakai, Kobayashi, Asai, Shimizu, Tokura, Tsurumaru, Matsui,
  Honjo, Tamaki, Takesue, Tokura, Dynes, Dixon, Sharpe, Yuan, Shields,
  Uchikoga, Legr\'{e}, Robyr, Trinkler, Monat, Page, Ribordy, Poppe, Allacher,
  Maurhart, L\"{a}nger, Peev, and Zeilinger}]{Sasaki}
Sasaki M, Fujiwara M, Ishizuka H, Klaus W, Wakui K, Takeoka M, Miki S,
  Yamashita T, Wang Z, Tanaka A, Yoshino K, Nambu Y, Takahashi S, Tajima A,
  Tomita A, Domeki T, Hasegawa T, Sakai Y, Kobayashi H, Asai T, Shimizu K,
  Tokura T, Tsurumaru T, Matsui M, Honjo T, Tamaki K, Takesue H, Tokura Y,
  Dynes JF, Dixon AR, Sharpe AW, Yuan ZL, Shields AJ, Uchikoga S, Legr\'{e} M,
  Robyr S, Trinkler P, Monat L, Page JB, Ribordy G, Poppe A, Allacher A,
  Maurhart O, L\"{a}nger T, Peev M, Zeilinger A (2011) Field test of quantum
  key distribution in the tokyo qkd network. Opt Express 19(11):10,387--10,409,
  \doi{10.1364/OE.19.010387}

\bibitem[{Scarani et~al(2009)Scarani, Bechmann-Pasquinucci, Cerf, Du\u{s}ek,
  L\"{u}tkenhaus, and Peev}]{RMP}
Scarani V, Bechmann-Pasquinucci H, Cerf NJ, Du\u{s}ek M, L\"{u}tkenhaus N, Peev
  M (2009) The security of practical quantum key distribution. {Rev Mod Phys}
  81:1301--1350

\bibitem[{Shoup(1996)}]{Shoup96}
Shoup V (1996) On fast and provably secure message authentication based on
  universal hashing. In: Koblitz N (ed) CRYPTO '96, Springer 1996, Lecture
  Notes in Computer Science, vol 1109, pp 313--328

\bibitem[{Stinson(1991)}]{Stinson92}
Stinson DR (1991) Universal hashing and authentication codes. In: Feigenbaum J
  (ed) CRYPTO '91, Springer 1992, Lecture Notes in Computer Science, vol 576,
  pp 74--85

\bibitem[{Taylor(1994)}]{Taylor}
Taylor R (1994) An integrity check value algorithm for stream ciphers. In:
  Stinson D (ed) Advances in Cryptology - CRYPTO '93, Springer-Verlag, Lecture
  Notes in Computer Science, vol 773, pp 40--48

\bibitem[{Wegman and Carter(1981)}]{WegmanCarter81}
Wegman MN, Carter JL (1981) J Comput Syst Sci 22:265

\end{thebibliography}

\begin{appendix}

\newcommand{\key}{s}
\newcommand{\randstring}{S}
\newcommand{\sequ}[2]{#1_1|#1_2|\dots|#1_{#2}}
%\newcommand{\operatorname}[1]{\mathrm{#1}}
%\newcommand{\binom}[2]{{#1 \choose #2}}
%\section{Algorithm (and probability) that generates a given key of length $M$
%from a random string of length $N$}

\section{Proof of Lemma \ref{lem:collision1}}\label{ProofCollision1}

\setcounter{lemma}{0}
\begin{lemma} 
Let us assume that $f$ maps all messages in $\setM$ randomly onto $\mZ$.
%(for a randomly choosen $m\in\setB$ and $\forall z\in\mZ: \Pr\{f(m)=z\} = |\mZ|^{-1}$),
Then the probability that at least one of the messages in $\setM$ is validated by the given tag $t=h_K(f(\mA))$ is
\begin{equation*}
\Pcoll =\Pr\left\{\exists \mE\in\setM: h_K(f(\mE))=t) \right\} > 1-\exp\left(-|\setM| |\mZ|^{-1}\right).
\end{equation*}
\end{lemma}

\begin{proof}
By assumption,
the probability that $f$ maps any (randomly chosen) message $m$ of $\setM$ onto any fixed value $z$ of $\mZ$ is $1/|\mZ|$:
\begin{equation}
m\in_R \setM, \forall z\in\mZ: \Pr\left\{f(m)=z\right\}=1/|\mZ|.
\end{equation}
Applying $h_K$ to $f(m)$ and $z$ in the argument of $\Pr$ (which potentially increases the value of the probability), setting $z=f(\mA)$, and using $t=h_K(f(\mA))$ we obtain 
\begin{equation}
m\in_R \setM: \Pr\left\{h_K(f(m))=t\right\}\ge 1/|\mZ|.
\end{equation}
Consequently, the probability that $t$ authenticates at least one
message of all $|\setM|$ different messages in $\setM$ is at least
$1-\left(1-|\mZ|^{-1}\right)^{|\setM|}$, and using that $(1-1/n)^{n} <
e^{-1}$ for $n>1$ finishes the proof.  
\end{proof}
If desired, $1/|\mZ|$ can be replaced by any lower bound on the
probability to allow for non-uniform distributions.

\section{Proof of Corollary \ref{cor:collision2}}\label{ProofCollision2}

\setcounter{corollary}{0}
\begin{corollary}
Let $\setB$ be the closed ball of all messages $m$ having a Hamming distance to $m^\Eve$ not exceeding $w$:
\begin{equation*}
\setB=\left\{ m: d_H(m,m^\Eve)\le w
\right\},
\end{equation*}
and let us assume that $f$ maps all messages in $\setB$ randomly onto $\mZ$.
%(for a randomly choosen $m\in\setB$ and $\forall z\in\mZ: \Pr\{f(m)=z\} = |\mZ|^{-1}$),
Then the probability that at least one of the messages in $\setB$ is validated by the given tag $t=h_K(f(\mA))$ is
\begin{equation*}
\Pcoll =\Pr\left\{\exists \tilde{m}^\Eve\in\setB: h_K(f(\tilde{m}^\Eve))=t) \right\} > 1-\exp\left(-|\setB| |\mZ|^{-1}\right).
\end{equation*}
For simplicity we can loosen the bound and replace $|\setB|$ by ${\ell \choose w} < |\setB|$, where $\ell$ is the length of the binary message $m^\Eve$.

% Under the assumption that a random message $m\in \mathcal{B}$
%chosen with uniform probability $|\mathcal{B}|^{-1}$
%has tag $t$ with probability $|\mZ|^{-1}|$, i.e.
%Assume that the collision probability of two distinct messages under $f$ does not depend on their Hamming distance (provided it is larger than some very small number).
\end{corollary}

\begin{proof}
The proof follows from Lemma \ref{lem:collision1} by setting $\setM=\setB$.
% By assumption,
% %
% the probability that $f$ maps any (randomly chosen) message $m$ of $\setB$ onto any fixed value $z$ of $\mZ$ is $1/|\mZ|$:
% \begin{equation}
% m\in_R \setB, \forall z\in\mZ: \Pr\left\{f(m)=z\right\}=1/|\mZ|.
% \end{equation}
% Applying $h_K$ to $f(m)$ and $z$ in the argument of $\Pr$ (which potentially increases the value of the probability), setting $z=f(\mA)$, and using $t=h_K(f(\mA))$ we obtain 
% \begin{equation}
% m\in_R \setB: \Pr\left\{h_K(f(m))=t\right\}\ge 1/|\mZ|.
% \end{equation}
% Consequently, the probability that $t$ authenticates at least one
% message of all $|\setB|$ different messages in $\setB$ is at least
% $1-\left(1-|\mZ|^{-1}\right)^{|\setB|}$, and using that $(1-1/n)^{n} <
% e^{-1}$ for $n>1$ finishes the proof.  

Finally,
$|\setB|=\sum_{k=0}^{w} {\ell \choose k}> {\ell \choose w}$.
%\begin{eqnarray*}
%\Pcoll =\Pr\left\{\exists \tilde{m}^\Eve\in\setB: f(\tilde{m}^\Eve)=z_0) \right \} > 1-\left(1-|\mZ|^{-1}\right)^{N_C}.
%\end{eqnarray*}
\end{proof}
If desired, $1/|\mZ|$ can be replaced by any lower bound on the
probability to allow for non-uniform distributions.

\section{Calculation of success probability in step (Pc') of attacks 2 and 3} \label{app:collprob2}

\subsection{Attack 2--Toeplitz based hashing}

The probability that step (Pc') is successful, i.e. that a message $\PE$ exists that fulfills $\PE(\SEB)=\KA$ and $t_2=\MAC{K_2}{\BE,\TE,\PE}$, depends on the length of the message $\PE$: a short message $\PE$ means less freedom for Eve to find collisions. 
% that selects the PA function to be used in the QKD protocol. 
% The length of this binary message will be (at least) the $\log_2$ of the number of different PA functions that can be chosen. 
%For PA so-called strong quantum-proof randomness extractors must be used\cite{Renner?}. 
Currently, owing to its low computational complexity universal hashing with Toeplitz matrices is predominantly used for PA. 
The smallest known Toeplitz matrix based hashing family $\mathcal{H}_T: \bin^n \rightarrow \bin^m$ consists of $2^{n-1}$ different functions \cite{Hayashi2011}. 
Note, that it has $2^m$ different images and that for any $x\ne 0$, $2^{n-m-1}$ different functions exist in $\mathcal{H}_T$ that map $x$ to any $y$. 

Thus we need $\len(\SEB)-1$ bit for the description of a particular function of this family (here $\len(\cdot)$ denotes the length of a binary string) and the size of the set $\setP$ of all PA functions that fulfill the first condition, i.e. $\setP:=\{\PE: \PE(\SEB)=\KA\}$, is given by $|\setP|=2^{\len(\SEB)-\len(\KA)-1}$.

%First we consider the set $\setP$ of all PA functions that fulfill the first condition, $\setP:=\{\PE: \PE(\SEB)=\KA\}$. 
% In total $2^{|\SEB|-|\KA|-1}$ functions map any fixed value $\SEB\ne 0$ to any fixed $\KA$. 
%Thus in the current protocol $|\SEB|-1$ bits are needed to select one function 
% Trevisan's extractor asymptotically has a There exist several variants \cite{Tsurumaru2013} with different parameters. 

Combining the fixed messages $\BE$ and $\TE$ with all messages from $\setP$ we form the set of all triples $\setM:=\{(\BE,\TE,\PE): \PE \in \setP\}$ that fulfill the first condition. Obviously, $|\setM| = |\setP|$. 
Now Lemma~\ref{lem:collision1} gives us a lower bound for the success probability for finding a message in $\setM$ that collides with $M^\Alice$, i.e. 
$\Pcoll =\Pr\left\{\exists \mE\in\setM: h_{K_2}(f(\mE))=t_2) \right\} > 1-\exp\left(-|\setM|/|\mZ|\right)$.

If we assume again that $|\mZ|=2^{256}$ this means that shrinking the corrected key $\SEB$ by 260 bit or more to obtain the final key $\KA$ gives Eve a chance of $1-\exp(-2^{260-1}/2^{256} =1 -\exp(-8) \ge 0.999$ to find a collision.

\subsection{Attack 3--Toeplitz based hashing}

For attack 3 the argumentation is completely analogous: the only difference is that we directly define $\setM:=\{\PE: \PE(\hatSB)=\KA\}$ and that $t_2$ is replaced with $t_5$.

\subsection{Discussion for other PA functions}

Besides universal hashing other strong randomness extractors can be used for PA. For example, asymptotically, Trevisan's extractor needs a shorter description (seed) to select a particular function. In such a case, it might be impossible for Eve to fulfill both conditions. Nevertheless, she can accept that she shares different keys with Alice and Bob and only search for a $\PE$ such that the MAC is accepted. In that case it would be only necessary that the seed is at least 260 bit long. To the best of our knowledge all strong extractors with useful parameters need much more than 260 bit of seed.

\section{Subsequence problem} \label{app:subsequence}

Eve is given two fixed bit sequences, $\SEB$ (sifted key that Eve wants to achieve) and $\IA$ (the raw key of Alice). Her goal is to find  a \emph{subsequence} of $\IA$ that coincides with $\SE$.

\subsection{Algorithm that finds a subsequence}\label{alg:bases1}

%In the attack, Eve will be given two fixed bit sequences, \mbox{$s=\sequ{s}{m}$} (the sifted key she wants to achieve, $s_E$) and \mbox{$S=\sequ{S}{n}$} (the key prepared by Alice, $d_A$), and her attack will succeed if $s$ is a \emph{subsequence} of $S$, denoted $s\triangleleft S$, or $S\triangleright s$.

First we give a simple algorithm that takes two
sequences \mbox{$s=\sequ{s}{m}$}, \mbox{$S=\sequ{S}{n}$} as inputs and returns
the index set $\mathcal{J}=\{j_1,\dots,j_m\}=\{j_i: S_{j_i}=s_i\}$ if $s$ is a subsequence of $S$ (denoted $s\preccurlyeq S$). \\
%\label{alg:bases1}
\noindent {\bf Algorithm A} find subsequence($s$, $S$)\\
Input: two non-empty binary sequences $\key$ and $\randstring$.\\
Output: index set $\mathcal{J}$ if $\key$ is a subsequence of $\randstring$, else $\emptyset$.
% \begin{enumerate}[1: ]
% \item $i\leftarrow 1$, $j\leftarrow 1$
% \item {\bf if} $\key_i=\randstring_j$ {\bf then} $i\leftarrow i+1   \qquad
% \qquad$ // we found one bit of $\key$
% \item $j\leftarrow j+1$
% \item {\bf if} $(i\le M$ and $j\le N)$ {\bf then goto} 2
% \item {\bf if} $i>M$ {\bf then return} true {\bf else return} false
% \end{enumerate}

\begin{tabbing}
\hspace*{1cm}\=\hspace*{1cm}\=\hspace*{4cm}\= \kill
 $i\leftarrow 1$, $j\leftarrow 1$, $m\leftarrow$ length$(s)$, $n\leftarrow$ length$(S)$, $\mathcal{J} \leftarrow \emptyset$ \\
 {\bf do} \\
 \> {\bf if} $\key_i=\randstring_j$ {\bf then} \> \> // we found one bit of $\key$ \\
 \> \> $\mathcal{J}\leftarrow \mathcal{J}\cup \{j\}$ \> // store position \\
 \> \> $i\leftarrow i+1$ \> // compare next bit of $s$ \\
 \>  {\bf endif} \\
 \>  $j\leftarrow j+1$ \> \> // compare next bit of $S$ \\
 {\bf while} $(i\le m$ and $j\le n)$ \> \> \> // neither end of $s$ nor end of $S$ reached \\
 {\bf if} $i\le m$ {\bf then return} $\emptyset$ {\bf endif} \> \> \> // end of $s$ not reached, but end of $S$ reached \\% $\Rightarrow s$ is not a subsequence of $S$\\
 {\bf return} $\mathcal{J}$ \\
\end{tabbing}

\subsection{Probability for finding a subsequence in a random sequence}

We assume that both sequences consist of i.i.d.\ Bernoulli trials with $p(0)=p(1)=1/2$ and
calculate the (success) probability that $s\preccurlyeq S$.

$s\preccurlyeq S$ iff $S$ is of the form
\begin{equation}
S=\bar s_1|\dots|\bar s_1|\boldsymbol{s_1}| \bar s_2|\dots|\bar s_2|\boldsymbol{s_2}| \dots| \bar s_m|\dots|\bar s_m|\boldsymbol{s_m}| x_1|x_2|\dots.
\end{equation}
Here, $\bar s_j$ denotes the negation of $s_j$ (written above as $\boldsymbol{s_j}$ to improve readability), while each $x_i$ can independently take value $0$ or $1$. All sequences $\bar s_j|\dots|\bar s_j$ are optional. Let $\mathsf{S}$ be the \emph{number of different} valid sequences, i.e. sequences $S$, that contain $s$ as a subsequence. Obviously $\mathsf{S}$ does not depend on $s$, but only on $m$ and $n$. To calculate $\mathsf{S}$ we can  therefore choose  $s$ to be the all zero sequence of length $m$. Consequently, $\mathsf{S}$ is equal to the number of different binary sequences of length $n$ that contain at least $m$ zeroes.
The success probability
\begin{eqnarray}
\operatorname{Prob}\{s\preccurlyeq S\} =\mathsf{S}/2^n=2^{-n}\sum_{l=m}^n \binom{n}{l}.
\end{eqnarray}

\subsection{Application to Eve's attack}
Note that Eve wants to find the sifted key $\SEB$ in Alice's raw key $\IA$. If both bases are used with equal probability (as in standard symmetric BB84), then $m\approx n/2$. Obviously,
\begin{eqnarray}
\operatorname{Prob}\{s\preccurlyeq S\} >\frac{1}{2} \Leftarrow m \le \lfloor n/2\rfloor.
\end{eqnarray}

However, it is not necessary, that $s$ is an \emph{exact} subsequence of $S$. We can allow for some errors that will be removed during the subsequent error correction step.
Using Hoeffding's inequality (Theorem 1 in Ref.~\cite{Hoeffding63}) we can give a non-tight (but exponential) lower bound on $\operatorname{Prob}\{\tilde{s}\preccurlyeq S\}$ if we allow for approximately $k$ errors in the resulting subsequence $\tilde{s}$:

\begin{equation}
\operatorname{Prob}\{\tilde{s}\preccurlyeq S\}\ge 1-\exp\left(-\frac{2k^2}{n}\right) \Leftarrow \tilde{m}=\lfloor n/2\rfloor-k.
\end{equation}
Here, only $\tilde{m}$ bits of $s$ form a subsequence of $S$.
For moderate values of $k$ this probability reaches almost unity.

\section{Universal hash functions and proof of Theorem~\ref{thm:fix}}
\label{uhf}

In the following we give definitions of (Almost) Universal$_2$ and
Strongly Universal$_2$ hash function families; see e.g.,
\cite{Stinson92}.
\begin{definition}[$\epsilon$-Almost Universal$_2$ ($\epsilon$-AU$_2$) 
  hash functions]
  Let $\mM$ and $\mT$ be finite sets. A family $\mH$ of hash functions
  from $\mM$ to $\mT$ is \emph{$\epsilon$-Almost Universal$_2$} if
  there exist at most $\epsilon|\mH|$ hash functions $h\in\mH$ such
  that $h(m_1) = h(m_2)$ for any two distinct $m_1, m_2 \in
  \mM$. \medskip

  \noindent If $\epsilon=1/|\mT|$, then $\mH$ is called
  \emph{Universal$_2$ (U$_2$)}.
\end{definition}

\begin{definition}[$\epsilon$-Almost Strongly Universal$_2$
  ($\epsilon$-ASU$_2$) hash functions]
  \label{def:ASU2}
  Let $\mM$ and $\mT$ be finite sets. A family $\mH$ of hash functions
  from $\mM$ to $\mT$ is $\epsilon$-ASU$_2$ if the following two
  conditions are satisfied:
  \begin{enumerate}[(a)]
  \item The number of hash functions in $\mH$ that takes
    an arbitrary $m_1 \in \mM$ to an arbitrary $t_1 \in \mT$ is
    exactly $|\mH|/|\mT|$.
  \item The fraction of those functions that also takes an
    arbitrary $m_2 \neq m_1$ in $\mM$ to an arbitrary $t_2 \in \mT$
    (possibly equal to $t_1$) is at most $\epsilon$.
  \end{enumerate}
  \noindent If $\epsilon=1/|\mT|$, then $\mH$ is called \emph{Strongly
    Universal$_2$ (SU$_2$)}.
\end{definition}

Below, we have restated Theorem~\ref{thm:fix} together with its
proof. This theorem states that the composition of a hash function
family with an SU$_2$ family will form an ASU$_2$ family if and only
if the first family in the composition is AU$_2$.  The ``if'' part
follows from the composition theorem \cite{Stinson92}, but the below proof
simultaneously handles ``if and only if''.

%%%%%%%%%%%%%%%%%%%%%%%%%%%%%%%%%%%%%%%%%%%%%%%%%%%%%%%%%%%%%%%%%%%%
%%%%%%% Theorem 3 in combinatorial form %%%%%%%%%%%%%%%%%%%%%%%%%%%%
%%%%%%%%%%%%%%%%%%%%%%%%%%%%%%%%%%%%%%%%%%%%%%%%%%%%%%%%%%%%%%%%%%%%
\setcounter{theorem}{0}
\begin{theorem}
  Let $\mM$, $\mZ$ and $\mT$ be finite sets. Let $\mF$ be a family of
  hash functions from $\mM$ to $\mZ$, $\mH$ a family of SU$_2$ hash
  functions from $\mZ$ to $\mT$, and $\mG := \mH\circ\mF$, where
  $\circ$ stands for element-wise composition. Then $\mG$ is
  $\epsilon$-ASU$_2$ if and only if $\mF$ is $\epsilon'$-AU$_2$, where
  $\epsilon=\epsilon'(1-1/|\mT|)+1/|\mT|$.
\end{theorem}
\begin{proof}
  For $\mG$ to be $\epsilon$-ASU$_2$, there are two requirements
  (Definition~\ref{def:ASU2}). The first, on $|\{g: g(m)=t\}|$, needs
  no properties of $\mF$, because, for any $m\in\mM$ and $t\in\mT$,
  \begin{equation}\label{eq:asu-cond-1}
  \begin{aligned}
  |\{g: g(m)=t\}| &= \sum_z|\{f:f(m)=z\}||\{h:h(z)=t\}|\\
                  &= \sum_z|\{f:f(m)=z\}|\frac{|\mH|}{|\mT|}
                   = |\mF|\frac{|\mH|}{|\mT|}=\frac{|\mG|}{|\mT|}.
  \end{aligned} 
  \end{equation}
  The second requirement is a bound for   
  \begin{equation}\label{eq:asu-cond-2}
    \begin{aligned}
      |\{g&:g(m_1)=t_1, g(m_2)=t_2\}|\\
      &= \sum_z |\{f:f(m_1)=f(m_2)=z\}||\{h:h(z)=t_1,h(z)=t_2\}|\\
      &\quad + \sum_{z_1\neq z_2} |\{f:f(m_1)=z_1,f(m_2)=z_2\}|
      |\{h:h(z_1)=t_1,h(z_2)=t_2\}|,
    \end{aligned}
  \end{equation}
  for any two distinct $m_1$, $m_2\in\mM$.  If $t_1\neq t_2$, the
  first term above will be zero because $h(z)$ will never equal both
  $t_1$ and $t_2$. If instead $t_1=t_2=t$, the first term simplifies
  to
  \begin{equation}\label{eq:1}
    \begin{aligned}
      \sum_z |\{f:f(m_1)=f(m_2)=z\}||\{h:h(z)=t\}| =
      |\{f:f(m_1)=f(m_2)\}|\frac{|\mH|}{|\mT|}.
    \end{aligned}
  \end{equation}  
  The second term is
  \begin{equation}\label{eq:2}
    \begin{aligned}
      \sum_{z_1\neq z_2} |\{f:f(m_1)=z_1,f(m_2)=z_2\}|
      |\{h&:h(z_1)=t_1,h(z_2)=t_2\}|\\
      &= \big(|\mF|-|\{f:f(m_1)=f(m_2)\}|\big)\frac{|\mH|}{|\mT|^2}
    \end{aligned}
  \end{equation}
  and this can be bounded by $|\mG|/|\mT|^2$ only using properties of
  $\mH$.  Thus, if $t_1=t_2$ the first term needs a bound for
  collisions within $\mF$, while the second only needs properties of
  $\mH$, and we obtain
  \begin{equation}\label{eq:asu-cond-2n}
    |\{g:g(m_1)=t_1, g(m_2)=t_2\}|
    = |\{f:f(m_1)=f(m_2)\}|\Big(\delta_{t_1,t_2}-\frac{1}{|\mT|}\Big)
    \frac{|\mH|}{|\mT|} + \frac{|\mG|}{|\mT|^2},
  \end{equation}
  where $\delta_{t_1,t_2}=1$ if $t_1=t_2$ and $0$ otherwise.  This
  makes the second requirement on $\mG$ equivalent to $\mF$ being
  $\epsilon'$-AU$_2$:
  \begin{equation}\label{eq:iff}
    \begin{aligned}
      |\{g:g(m_1)=t_1, g(m_2)=t_2\}| &\le \epsilon\frac{|\mG|}{|\mT|}=
      \epsilon'\Big(1-\frac{1}{|\mT|}\Big)
      \frac{|\mG|}{|\mT|} + \frac{|\mG|}{|\mT|^2}\\
      {\centering \Longleftrightarrow} & \\
      |\{f:f(m_1)=f(m_2)\}| &\le \epsilon'|\mF|.
    \end{aligned}
  \end{equation} 
\end{proof}
\end{appendix}

%%% Local Variables: 
%%% mode: latex
%%% TeX-master: "NJPmain"
%%% End: 

\end{document}